\makeatletter
\newcommand*{\rom}[1]{\expandafter\@slowromancap\romannumeral #1@}
\makeatother

\documentclass[journal,12pt,onecolumn,draftclsnofoot]{IEEEtran}

\IEEEoverridecommandlockouts
\ifCLASSINFOpdf
\else
\fi
\usepackage{graphics}
\usepackage[table]{xcolor}

\usepackage{booktabs}
\usepackage{tabularx}
\usepackage{makecell}
\usepackage{pifont}
\newcommand{\cmark}{\ding{51}}
\newcommand{\xmark}{\ding{55}}

\newenvironment{sproof}[1][Sketch of proof]{%
  \noindent\emph{#1.}\quad
}{%
  \qed\bigskip
}
\usepackage{units}

\usepackage{enumitem}
\usepackage{overpic}
\usepackage{amsthm}

\usepackage{multicol}
\usepackage{lipsum}
\usepackage{color}
\usepackage[cmex10]{amsmath}
\usepackage{amsthm}
\usepackage{amssymb}
\usepackage{mathrsfs}
\usepackage{mathtools}
\usepackage{amsbsy}
\usepackage{xcolor}
\usepackage{mathrsfs}
\usepackage{setspace}
\usepackage{float}
\usepackage{graphicx}
\usepackage{psfrag}
\usepackage{epstopdf}
\usepackage{lettrine}
\usepackage{psfrag}
\usepackage{newclude}
\usepackage[normalem]{ulem}
\usepackage{latexsym}
\usepackage{algpseudocode}
\usepackage{algorithm,algpseudocode}
\usepackage{algorithmicx}

\usepackage{multirow}
\usepackage{rotating}
\usepackage{booktabs}
\usepackage{amsmath}
\usepackage{wasysym}
\usepackage{mathrsfs}
\usepackage{bbm} 
\usepackage{filecontents}
\usepackage{amsmath,epsfig,cite,amsfonts,amssymb,psfrag,color}

\usepackage{apptools}

\usepackage{chngcntr}
\ifCLASSOPTIONcompsoc
\usepackage[caption=false,font=normalsize,labelfon
t=sf,textfont=sf]{subfig}
\else
\usepackage[caption=false,font=footnotesize]{subfig}
\fi
\AtAppendix{\counterwithin{lemma}{section}}

\usepackage{accents}
\newlength{\dhatheight}

\allowdisplaybreaks
\theoremstyle{remark}
\newtheorem{remark}{Remark}
\theoremstyle{example}
\newtheorem{example}{Example}
\theoremstyle{theorem}
\newtheorem{theorem}{Theorem}
\newtheorem{lemma}{Lemma}

\newtheorem*{lemma*}{Lemma} 
\newtheorem{lemalpha}{Lemma}

\newtheorem*{example*}{Example} 
\graphicspath{{figures/}}
\allowdisplaybreaks

\usepackage[skins,theorems]{tcolorbox}
\tcbset{highlight math style={shrink tight,extrude by=0.55mm,colframe=red,
		boxrule=0.4pt,frame style={opacity=0.6}}}

\theoremstyle{definition}

\begin{document}
	\title{Degrees of Freedom of Cache-Aided Interference Channels
Assisted by Active Intelligent Reflecting Surfaces}
\author{Abolfazl Changizi, Ali H. Abdollahi Bafghi, Mahtab Mirmohseni, and Masoumeh Nasiri-Kenari}



\maketitle

	\textcolor[rgb]{0,0,0}{
		\begin{abstract}
			This paper studies cache-aided wireless networks in the presence of active intelligent reflecting surfaces (IRSs) from an information-theoretic perspective. Specifically, we investigate interference management in a cache-aided wireless network assisted by an active IRS to enhance the achievable degrees of freedom (DoF). To this end, we jointly design the content placement, delivery phase, and IRS coefficients, and propose a one-shot achievability scheme. Our scheme exploits transmitters' cooperation, cache contents, interference alignment, and IRS capabilities, based on the network parameters. We derive the achievable one-shot sum-DoF for different cache sizes, network configurations, and numbers of IRS elements, followed by an upper bound. Our results highlight the potential of deploying an IRS in cache-aided wireless communication systems. In particular, they underscore the enhancement of achievable DoF for various parameter regimes, especially when cache sizes are inadequate. Notably, we show that access to an IRS with a sufficient number of elements enables the achievement of the maximum possible DoF for various parameter regimes of interest.
	\end{abstract}}
	
	\IEEEpeerreviewmaketitle
	
	
	\section{Introduction}\label{sec1introduc}
	Future communication networks are expected to offer better data rates, connectivity, and high-quality services. 
    Intelligent reflecting surfaces (IRSs), also known as reconfigurable intelligent surfaces (RISs), enable the capability of reconfiguring wireless channels and are considered a promising solution in 6G  \cite{IRS_Survey2, IRS_Survey1, comprehensive_6g}. Specifically, 
    IRSs provide an opportunity to manipulate the electromagnetic properties of the incident waves in real-time, thereby offering improvements in 
    coverage extension \cite{IRS_Coverage}, localization accuracy \cite{Nasri_localization}, physical layer security enhancement \cite{Physical_layer_IRS_Security}, 
    and interference mitigation \cite{bafghi_TCom}. 

    One significant capability of IRSs is particularly their ability to weaken the cross channels that carry interference, thus greatly enhancing the degrees of freedom (DoF) of interference channels. For instance, in \cite{bafghi_TCom}, the authors showed that the assistance of an active IRS with at least $K(K-1)$ elements or a passive IRS with a sufficiently large number of elements can increase the sum-DoF‌ of a $K$-user time-selective interference channel, traditionally known for achieving a sum-DoF of $\frac{K}{2}$\cite{DoF_Jafar}, to $K$. In the same channel model, the authors of \cite{Yu_nulling_interference} proposed an algorithm to achieve a sum-DoF of $K$ using a passive IRS with $2K(K-1)$ elements, under the assumptions of blocked direct links and line-of-sight IRS–transmitter/receiver channels.
    It was also shown that the achievable sum-DoF of $\frac{K_T K_R}{K_T+K_R-1}$ for an X-network of size $K_T \times K_R$, derived in \cite{jafar_Xchannel}, can be increased to  $\min\{K_T, K_R\}$ using either a passive or active IRS with a sufficiently large number of elements \cite{bafghi_xnetwork}.

The DoF of multiple-input multiple-output (MIMO) interference channels assisted by IRSs has also been studied in the literature  
\textcolor{black}{\cite{MIMO_D2D, DoF_Modulation_IRS, zheng2022dof, K_User_MIMO, Rank_deficient_MIMO, jiang2025dofanalysisbeamformingdesign, RIS_Cooperative_IA_MIMO, Iterative_IA, noncoherentmimo}}. 
For instance, \cite{MIMO_D2D} studied a multi-antenna D2D network with a passive IRS and applied techniques based on Riemannian Manifold for a rank-minimization problem. In \cite{DoF_Modulation_IRS}, the authors showed that the sum-DoF of a MIMO interference channel with $K_T$ antennas at the transmitter and $K_R$ antennas at the receiver can be enhanced from $\min \{K_T, K_R\}$, derived in
  \cite{JafarMIMO}, to   $\min \{K_T+\frac{Q}{2}-\frac{1}{2}, Q, K_R\}$ using a $Q$-element IRS. \textcolor{black}{In \cite{zheng2022dof}, the authors derived achievable sum-DoF results and active IRS gain conditions for $2 \times 2$ MIMO interference channels with arbitrary antenna configurations.
 In IRS-assisted K-user MIMO networks, \cite{K_User_MIMO} showed that joint active-passive beamforming based on interference subspace alignment improves the DoF compared with the case without an IRS, while \cite{Rank_deficient_MIMO, jiang2025dofanalysisbeamformingdesign} derived lower and upper DoF bounds for rank-deficient channels.} 
 \textcolor{black}{In \cite{RIS_Cooperative_IA_MIMO}, a passive IRS-assisted cooperative interference alignment scheme was proposed for MIMO multi-user networks, leveraging space–time precoding to enhance achievable DoF with fewer IRS elements.
 Additionally, \cite{noncoherentmimo} studied scenarios without channel state information (CSI) in non-coherent settings.}
 Furthermore, active IRSs have been investigated in the presence of security constraints, where properly designed IRS beamforming can suppress information leakage 
  and increase the secure DoF of the system \cite{Yener_Wiretap, SDoF_Broadcast_Confidential, SDoF_MIMO_Wiretap_Su, SDoF_Green}.

 On the other hand, the exponential
growth and non-uniform distribution of network traffic cause  congestion during peak-traffic periods and under-utilization during off-peak times in the network \cite{main_caching_maddah_niesen}.
Caching in distributed memories across the base stations
and end users is an efficient approach to mitigate congestion, reduce
backhaul cost, and manage physical layer interference  \cite{main_caching_maddah_niesen, nader}. Wireless caching has been explored 
in various networks, including heterogeneous networks \cite{hetnet_cache}, 
D2D networks \cite{cache_D2D}, cloud-radio access networks \cite{content_tao, cache_bafghi, cache_abarghooyi}, and adversarial channels \cite{zamani2025cache}. 
In the following, we provide an overview of the information-theoretic literature on cache-aided wireless networks, with a particular emphasis on cache-aided interference channels. 

An information-theoretic approach toward caching, also known as coded caching, has been initially introduced by Maddah-Ali and Niesen in an error-free broadcast channel (BC) with one server and multiple cache-enabled users \cite{main_caching_maddah_niesen}. Their work has shown that cache memories at the receivers can offer significant global gains, in addition to the existing local caching gains. The authors then extended coded caching to a distributed network, achieving a close-to-optimal rate compared to centralized methods \cite{maddah2014decentralized}. In \cite{shariatpanahi2016multi}, coded caching has been studied in scenarios with multiple servers, resulting in significant performance enhancements. 

Caching at transmitters has also been explored in \cite{Madda_niesen_interfernce_33, simeone, niesen_caching_DoF, nader, naderializadeh2019cache, Tao_Fundamental, Tao_partial}, enabling effective interference management, thanks to the potential for cooperation among transmitters. The idea was first introduced in a $3 \times 3$ channel, where only transmitters had cache memories \cite{Madda_niesen_interfernce_33}. Subsequently, it was extended to any number of cache-enabled transmitters and receivers in \cite{simeone}. Cache-aided interference channels, where both transmitters and receivers are equipped with cache memories, have been also studied in high signal-to-noise ratio (SNR) regimes, with DoF or normalized delivery time (NDT) as the performance metric \cite{niesen_caching_DoF, nader, naderializadeh2019cache, Tao_Fundamental, Tao_partial}. In \cite{niesen_caching_DoF}, a coded caching scheme was investigated, based on the separation of the physical layer and network layer. In \cite{nader}, linear one-shot schemes were proposed, achieving a sum-DoF within a constant factor of $2$ optimality in a general network setting with caches at both transmitters and receivers. The authors in \cite{naderializadeh2019cache} studied this problem in a wireless cellular network 
under one-shot linear schemes. In \cite{Tao_Fundamental}, lower and upper bounds on the NDT were derived, with the proposed scheme achieving optimality in certain cache size regimes and a bounded multiplicative gap elsewhere. The idea of coded caching was also explored in 
partially-connected networks, including linear networks \cite{Tao_partial, shariatpanahi2016multi}, random networks \cite{random_topology}, and flexible networks \cite{shariatpanahi2016multi}.

The distinct potentials of IRSs and caching for interference management naturally raise the following question: Can their simultaneous use provide synergistic gains beyond what either technique can achieve alone? To address this question, we study a cache-aided interference channel, where cache memories are deployed both at the transmitters and receivers, in the presence of an active
IRS. We aim to derive the achievable sum-DoF of this channel by exploiting transmitters’ cooperation, cache contents (as side information), interference alignment, and IRS capabilities, \textcolor{black}{based on} the network parameters. 

For the system under study, we summarize the main contributions of this paper as follows:
\begin{itemize}
    \item First, we introduce a general wireless interference channel with caches deployed at all transmitters and receivers assisted by an active IRS. The IRS enables the elimination of undesired cross-links, while caching offers the potential for transmitters' cooperation and the utilization of cached data as side information for interference management.
    \item Then, we propose a one-shot achievability scheme with a closed-form sum-DoF of the network, followed by the derivation of an upper bound. To this end, we first present our results for the scenario where there is no overlap between the caches of transmitters, i.e., when the cache size at the transmitters is minimum, and for any cache sizes at receivers. Next, we extend the results for larger caches at transmitters and any size of caches at receivers within certain network sizes.
    In that regard, we meticulously design the cache placement and delivery phase, incorporating carefully chosen IRS coefficients, beamforming coefficients, and packet scheduling.
    \item Finally, numerical results are provided to compare the performance of the proposed schemes with the seminal works that have utilized caching for interference management. Our results show that the proposed schemes can notably increase the sum-DoF, especially when the cache sizes are inadequate at transmitters. 
    It also demonstrates that the maximum possible sum-DoF for many parameter regimes of interest can be achieved with a sufficient number of IRS elements, which was not achievable in the schemes without IRSs.
    
\end{itemize}

The most closely related work to our study is \cite{caching_IRS_Caire}, which considers a passive IRS-assisted network consisting of a \textcolor{black}{single} multiple-antenna server and \textcolor{black}{multiple} single-antenna cache-aided users
and proposed a grouping algorithm to address a combinatorial optimization problem to achieve the maximum DoF. 
The main differences between our work and \cite{caching_IRS_Caire} are as follows. First, we study an active IRS, whereas \cite{caching_IRS_Caire} considers a passive IRS, leading to different 
interference-nulling mechanisms and IRS size requirements. Second, the system models differ fundamentally, i.e., \cite{caching_IRS_Caire} focuses on a  multiple-input single-output (MISO) broadcast channel with no direct links, while we consider a $K_T \times K_R$ single-antenna interference channel allowing for the presence of direct links. Third, in contrast to \cite{caching_IRS_Caire}, which assumes full library access at the transmitter, we consider transmitter-side caching with arbitrary overlap determined by the cache size.  Finally, the coded caching schemes are different, with \cite{caching_IRS_Caire} relying on multiple-antenna placement delivery arrays-based coded caching \cite{MAPDA}, while our approach builds on the caching framework in \cite{nader} and incorporates the IRS into the system using techniques from \cite{bafghi_TCom, bafghi_xnetwork}. Overall, due to these differences,
the two works address complementary but distinct problems and fill different gaps in the literature.

\color{black}

A comparison of our work and the existing works on the DoF of interference channels assisted by IRSs can be found in Table \ref{tab:comparison}.

\begin{table*}[t]
\centering
\caption{Comparison of this work with related works on the DoF of interference channels assisted by IRSs.}
\label{tab:comparison}

\setlength{\tabcolsep}{4pt}
\renewcommand{\arraystretch}{1.15}
\footnotesize

\begin{tabular}{c cc cc cc c c cc p{2.9cm}}
\toprule
\textbf{Ref.} &
\multicolumn{2}{c}{\textbf{Caching}} &
\multicolumn{2}{c}{\textbf{IRS}} &
\multicolumn{2}{c}{\textbf{Multi-ant.}} &
\textbf{Scheme} &
\textbf{Security} &
\textbf{Lower} &
\textbf{Upper} &
\textbf{Channel} \\
\cmidrule(lr){2-3}\cmidrule(lr){4-5}\cmidrule(lr){6-7}
& \makecell{\textbf{Tx}} & \makecell{\textbf{Rx}}
& \makecell{\textbf{Active}} & \makecell{\textbf{Passive}}
& \makecell{\textbf{Tx}} & \makecell{\textbf{Rx}}
& & & & & \\
\midrule

\cite{bafghi_TCom} 
& \xmark & \xmark & \cmark & \cmark & \xmark & \xmark
& SE & \xmark & \cmark & \cmark & K-user IC \\

\cite{Yu_nulling_interference} 
& \xmark & \xmark & \xmark & \cmark & \xmark & \xmark
& OS & \xmark & \cmark & \xmark & K-user IC \\

\cite{bafghi_xnetwork} 
& \xmark & \xmark & \cmark & \cmark & \xmark & \xmark
& SE & \xmark & \cmark & \cmark & X-network \\

\cite{MIMO_D2D}
& \xmark & \xmark & \xmark & \cmark & \cmark & \cmark
& SE & \xmark & \cmark & \xmark & MIMO D2D IC \\

\cite{DoF_Modulation_IRS} 
& \xmark & \xmark & \xmark & \cmark & \cmark & \cmark
& SE & \xmark & \cmark & \cmark & MIMO IC \\

\cite{zheng2022dof} 
& \xmark & \xmark & \xmark & \cmark & \cmark & \cmark
& OS & \xmark & \cmark & \xmark & $2 \times 2$ MIMO IC \\

\cite{K_User_MIMO} 
& \xmark & \xmark & \cmark & \cmark & \cmark & \cmark
& OS & \xmark & \cmark & \xmark & K-user MIMO IC \\

\cite{Rank_deficient_MIMO} 
& \xmark & \xmark & \cmark & \xmark & \cmark & \cmark
& OS & \xmark & \cmark & \cmark & K-user MIMO IC \\

\cite{jiang2025dofanalysisbeamformingdesign} 
& \xmark & \xmark & \cmark &  \xmark   & \cmark & \cmark 
& OS & \xmark & \cmark & \cmark  & K-user MIMO IC\\

\cite{RIS_Cooperative_IA_MIMO} 
& \xmark & \xmark & \xmark & \cmark & \cmark & \cmark
& SE & \xmark & \cmark & \xmark & K-user MIMO IC  \\

\cite{Iterative_IA} 
& \xmark & \xmark & \cmark & \xmark  & \cmark  & \cmark 
& OS & \xmark & \cmark & \xmark &  K-user MIMO IC\\

\cite{noncoherentmimo} 
& \xmark & \xmark & \xmark & \cmark & \cmark & \cmark
& SE & \xmark & \cmark & \xmark & NC MIMO IC\\

\cite{Yener_Wiretap}
& \xmark & \xmark & \cmark & \xmark & \cmark & \cmark
& OS & \cmark & \cmark & \cmark & MIMO WIC \\

\cite{SDoF_Broadcast_Confidential} 
& \xmark & \xmark & \cmark & \xmark & \cmark & \cmark
& OS & \cmark & \cmark & \cmark & 2-user MIMO BC \\

\cite{SDoF_MIMO_Wiretap_Su} 
& \xmark & \xmark & \cmark & \xmark & \cmark & \cmark
& OS & \cmark & \cmark & \cmark & 2-user MIMO WIC \\

\cite{SDoF_Green} 
& \xmark & \xmark & \xmark & \cmark & \cmark & \cmark 
& SE & \cmark & \cmark & \cmark & MIMO WIC \\

\cite{caching_IRS_Caire} 
& \xmark & \cmark & \xmark & \cmark & \cmark & \xmark
& OS & \xmark & \cmark & \cmark & K-user MISO BC \\

\rowcolor{gray!15} {Ours}
& \cmark & \cmark & \cmark & \xmark & \xmark & \xmark
& {OS} & \xmark & \cmark & \cmark & $K_T \times K_R$ IC \\
\bottomrule
\end{tabular}

\vspace{2mm}
{\footnotesize
SE: symbol-extension-based; OS: one-shot; W: Wiretap; IC: interference channel; NC: non-coherent.
}
\end{table*}

\textit{Organization}: The remainder of this paper is organized as follows. Section \ref{sec2sysmodel} presents the system model and problem formulation. \textcolor{black}{Section \ref{Main_Results_and_Discussions} presents the main results and discusses their implications. This is  followed by the achievability scheme in Section \ref{achivability_Schemes:sec}.} 
Numerical results illustrating the achievability scheme are provided in Section \ref{sec4Ssim}. Finally, Section \ref{sec5con} concludes the paper.
	
	\textit{Notations}: Symbols $\mathbb{C}$, $\mathbb{N}$, $\mathbb{R}$, and $\mathbb{R}^{+}$ denote the sets of complex, natural, real, and positive real numbers, respectively. We use $[a, b]$ to denote the real numbers between $a, b \in \mathbb{R}$. For $0 \leq k \leq n$, we denote the binomial coefficient with $n \choose k$. Sets and vector spaces are denoted by calligraphic uppercase letters. For a set $\mathcal{A}$, the cardinality, supremum, and infimum are denoted by $|\mathcal{A}|$, $\operatorname{sup}(\mathcal{A})$, and $\operatorname{inf}(\mathcal{A})$, respectively. We show the set $\{1, 2, \ldots, K\}$ by $[K]$ and the set $\{1, \ldots, i-1, i+1, \ldots, K\}$ by $[K]\backslash \{i\}$. The set  $\{i, i+1, \ldots, j\}$ for $j > i$ is shown by $[i:j]$ and the set $\{i, \ldots, k-1, k+1, \ldots, j\}$ for $i < k < j$ is denoted by $[i:j]\backslash\{k\}$. 
    Vectors and matrices are denoted by bold lower-case and upper-case letters, respectively. For a matrix $\mathbf{V}$, $v_{i, j}$ demonstrates the element in the $i$-th row and the $j$-th column. 
    We represent the expected value of a random variable $X$ by $\mathbb{E}[X]$ and the probability of an event $E$ by $\operatorname{Pr}(E)$. Furthermore, $X \sim \mathcal{C} \mathcal{N}(\mu,\sigma^2)$ indicates that $X$ is a complex Gaussian random variable with mean $\mu$ and variance $\sigma^2$. We use $\mathcal{L}(\alpha, \beta)$ to denote the linear combination of $\alpha$ and $\beta$. We also denote the indicator function as $\mathbb{I}(x)$, where $\mathbb{I}(x) = 0$ if $x = 0$, otherwise, $\mathbb{I}(x) = 1$. For the integers $i, j,$ and $m$, we denote $i \oplus_{m} j$ as:
    $
    	i \oplus_m j= 1 + (i+j-1  \bmod m),
    $
    where $\bmod$ is the modulo operation, returning the remainder of a division. Moreover, a function $f(\rho)$ is $o(\log (\rho))$, if $\lim _{\rho \rightarrow \infty} \frac{|f(\rho)|}{\log (\rho)}=0$, \textcolor{black}{and is $O(\rho)$ if there exist constants $c > 0$ and $\rho_0$ such that $|f(\rho)| \le c\,\rho,  \text{ for all } \rho \ge \rho_0$.} 

	\section{System Model and \textcolor{black}{Preliminaries}}\label{sec2sysmodel}
    \subsection{System Model}\color{black}\label{systemmodelsubsection}
	\begin{figure}
		\centering
		\begin{overpic}[scale=.28]{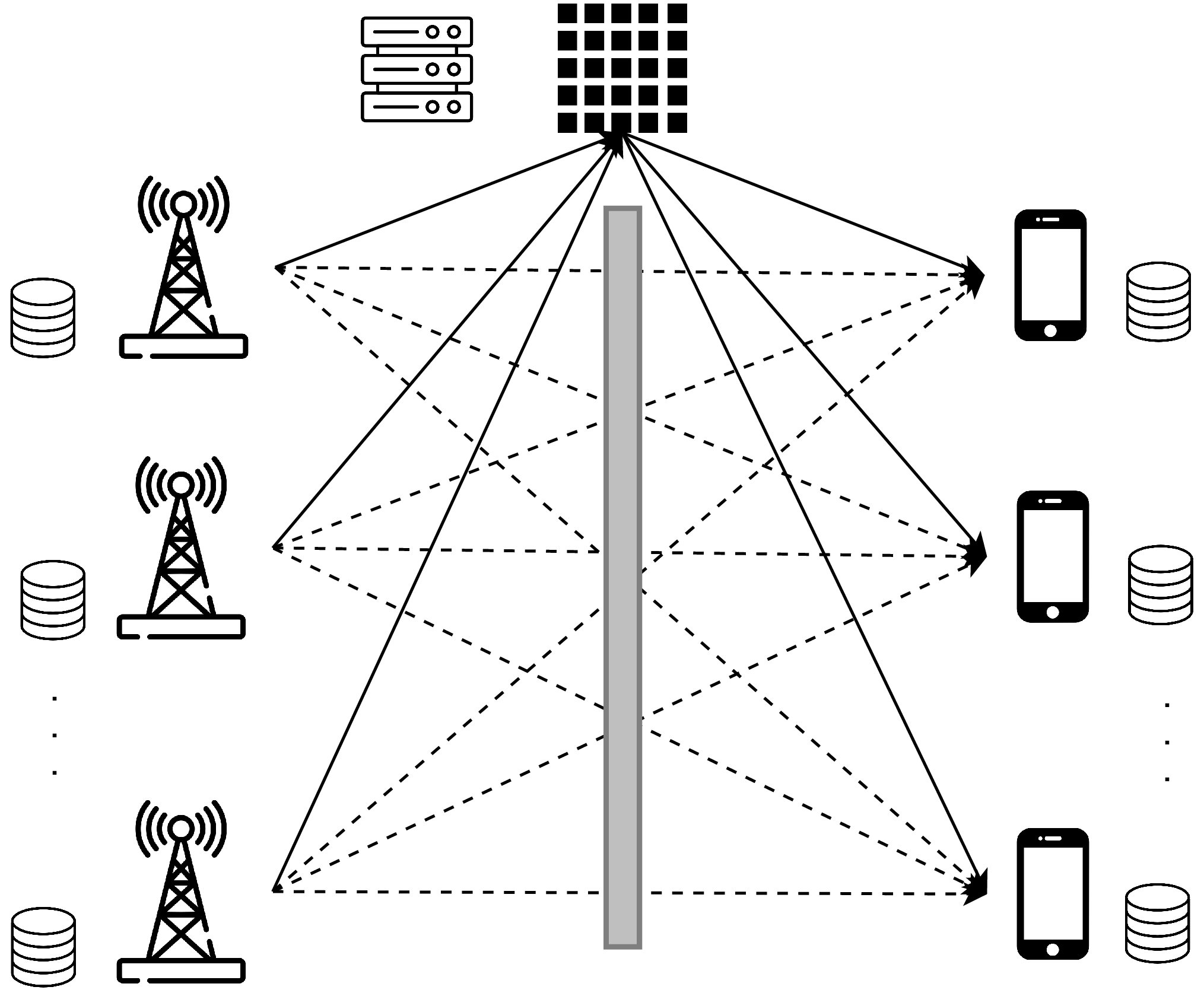}
			\put(1,65){\scriptsize $M_T$}
			\put(1.25,61){\scriptsize files}
			\put(1,41){\scriptsize $M_T$}
			\put(1.25,37){\scriptsize files}
			\put(1,13){\scriptsize $M_T$}
			\put(1.25,9){\scriptsize files}
			\put(93.5,66){\scriptsize $M_R$}
			\put(94,62){\scriptsize files}
			\put(93.5,42){\scriptsize $M_R$}
			\put(94,38){\scriptsize files}
			\put(93.5,14){\scriptsize $M_R$}
			\put(94,10){\scriptsize files}
			\put(10,75){\footnotesize Library of}
			\put(11.5,70){\footnotesize $N$ files}
			\put(39,-1){\footnotesize Smooth Blockage}
			\put(46.5,83){\footnotesize IRS (Q)}
			\put(13,48){\footnotesize Tx$_1$}
			\put(13,24){\footnotesize Tx$_2$}
			\put(13,-4){\footnotesize Tx$_{K_{T}}$}
			\put(85,49){\footnotesize Rx$_1$}
			\put(85,25){\footnotesize Rx$_2$}
			\put(85,-3){\footnotesize Rx$_{K_{R}}$}
			\put(65,66.5){\footnotesize $\mathbf{H}_{\mathrm{IR}}$}
			\put(32,67){\footnotesize $\mathbf{H}_{\mathrm{IT}}$}
			\put(65,59.5){\footnotesize $\mathbf{H}$}
		\end{overpic}
		\caption{Illustration of a cache-aided interference channel assisted by an active IRS with $Q$ elements, $K_T$ transmitters, and $K_R$ receivers, where each transmitter and receiver caches up to $M_T$ and $M_R$ files, respectively, from a library composed of $N$ files. The gray rectangle indicates a smooth blockage in the direct link, resulting in attenuation of the line-of-sight signals.}
	\label{Fig:System_Model}
	\end{figure}
	As illustrated in Fig. \ref{Fig:System_Model}, we consider a time-selective\footnote{Channel coefficients in different time slots
		are independent.} cache-aided interference channel assisted by a $Q$-element IRS, with $K_T$ single-antenna transmitters denoted by 
	$\left\{\mathrm{Tx}_i\right\}_{i=1}^{K_T}$ and $K_R$ single-antenna
	receivers denoted by 
	$\left\{\mathrm{Rx}_j\right\}_{j=1}^{K_R}$. The received signal at the $j$-th receiver in the $t$-th time slot is denoted by $Y^{[j]}(t)$, which is as follows:
    	
    \begin{align}\label{ch4}
		Y^{[j]}(t) 
        =  \sum_{i=1}^{K_T} \Bigl(\underbrace{ {H^{[j i]}(t) + \sum_{u=1}^{Q} H_{\mathrm{IR}}^{[j u]}(t) q^{[u]}(t) H_{\mathrm{TI}}^{[u i]}(t) }}_{\tilde{H}^{[j i](t)}}\Bigl)X^{[i]}(t)
        + Z^{[j]}(t),
	\end{align}where $X^{[i]}(t)$ indicates the transmitted signal of the $i$-th
	transmitter, 
    $H^{[j i]}(t)$ is the direct channel coefficient between the $i$-th transmitter and the $j$-th receiver, $H_{\mathrm{TI}}^{[u i]}(t)$ is the channel coefficient between the $i$-th transmitter and the $u$-th IRS‌ element, $H_{\mathrm{IR}}^{[j u]}(t)$ is the channel coefficient between the $u$-th IRS‌ element and the $j$-th receiver, and $Z^{[j]}(t)$ is the additive white Gaussian noise component with a variance of $1$ at the $j$-th receiver in the $t$-th time slot. 	
    Moreover, $q^{[u]}(t)$ is the coefficient applied by the $u$-th IRS element, written as 
	\begin{equation}\label{ch3}
q^{[u]}(t) =\rho^{[u]}(t) e^{j \phi^{[u]}(t)},\quad u \in\left[Q\right],
	\end{equation}
	where 
	$\phi^{[u]}(t) \in[0,2 \pi]$ and
	$\rho^{[u]}(t) \in \mathbb{R}^{+}$
	are the phase shift and amplitude factor applied to the received signal. Furthermore, 
    $\tilde{H}^{[j i]}(t)$ denotes the equivalent channel at time slot $t$. Additionally, the transmitted signal at the $i$-th transmitter is subject to the power constraint denoted by 
	$\mathbb{E}\left[\left|X^{[i]}(t)\right|^2\right] \leq P$.
	\begin{remark}\label{remark 1}
        
        This work focuses on active IRS‌s.
        In particular, 
        conventional passive IRSs suffer from a fundamental performance bottleneck caused by the multiplicative (or double-fading) effect, 
        which leads to severely attenuated reflected signals and negligible capacity gains when the direct links are not weak \cite{active_IRS, active_IRS_3}. Although this issue can be mitigated by increasing the number of reflecting elements or placing the IRS close to the transmitters/receivers, such solutions lead to large surface sizes, higher circuit power consumption, increased channel estimation overhead, and greater real-time optimization complexity, limiting practical deployment\cite{active_IRS_2, active_IRS_3}. In contrast, active IRSs,
        where reflecting elements are equipped with active loads (e.g., negative resistance) to amplify the incident signal at the electromagnetic level, 
        compensate for the double-fading attenuation \cite{active_IRS, active_IRS_2}. By enabling signal amplification in addition to passive beamforming, active IRSs can achieve strengthened IRS-aided links with fewer reflecting elements and reduced physical size under a given power budget, while avoiding the use of complex and power-hungry RF chains and remaining more cost- and energy-efficient than conventional active relays \cite{active_IRS_2, active_IRS_3}. Accordingly, in our system model, as the direct transmitter–receiver links are allowed to exist, the active IRS is deployed to enhance system performance.
   

	\end{remark}
      \color{black}

	The \textcolor{black}{active} IRS-assisted channel can mitigate unwanted cross-links between transmitters and receivers by adjusting the phase shift and amplitude of IRS elements, a capability unavailable in conventional wireless networks \cite{bafghi_TCom, bafghi_xnetwork}. To delineate the difference, we define the $K_T \times K_R$ network matrix $\mathcal{N}$, which characterizes the network connectivity using binary elements $n_{i,j}$. Specifically, $n_{i,j} = 1$ indicates a non-zero equivalent channel coefficient between the $i$-th transmitter and the $j$-th receiver, while $n_{i,j} = 0$ denotes a zero equivalent channel coefficient. This can be written as 
		$n_{i, j}=\mathbb{I}\bigl(\tilde{H}^{[j i](t)}\bigl)$.

	In each time slot $t$, the channel coefficients for direct links $(\text{i.e., }H^{[ji]}(t), \forall i, j)$ and the concatenated transmitter-IRS and IRS-receiver $(\text{i.e., }H_{\mathrm{TI}}^{[ui]}\left(t\right) H_{\mathrm{IR}}^{[j u]}\left(t\right), \forall i, j, u)$ of that time slot are all known at the transmitters, receivers, and IRS\footnote{We note that the DoF analysis under imperfect CSI is significantly more challenging than that under perfect CSI and typically requires substantially more complicated analytical tools (see, e.g., \cite{caching_imperfect_CSIT} for cache-aided settings and \cite{Arash1, Arash2} for systems without caching). Extending such analyses to IRS-assisted systems remains largely unexplored. Consequently, the DoF analysis with imperfect CSI 
represents an important direction for future research.}. Similar to \cite{bafghi_TCom, bafghi_xnetwork}, all channel coefficients $H^{[j i]}(t)$, $H_{\mathrm{IR}}^{[j u]}(t)$, and $H_{\mathrm{TI}}^{[u i]}(t)$ 
	are considered independent random variables for each
	$i$,
	$j$,
	$u$,
	and
	$t$, drawn from a continuous cumulative probability distribution. To ensure this independence, IRS elements must be spaced by more than half a wavelength, and the transmitter-to-IRS and IRS-to-receiver channels should be rich scattering \cite{bjornson2020rayleigh, bafghi_TCom}. 
	\begin{remark}
		When the spacing between IRS elements exceeds half the wavelength, channel correlations become negligible and can be considered independent \cite{IRS_applications, MISO_Schober, RIS_energy_efficiency, bafghi_TCom}. However, spacing beyond $\frac{\sqrt{2}}{2}$-wavelength  significantly degrades performance due to inefficient aperture use and grating lobes \cite{phased_array_paradox, bafghi_TCom}. This creates a trade-off. This work assumes half-wavelength spacing for simplicity,  leading to negligible channel correlations and grating lobes.  
	\end{remark}
	
	We assume that each receiver requests an arbitrary file from a library comprising $N$ files $\mathcal{W}=\left\{W_{k}: k \in[N]\right\}$.
	Each file $W_{k}$ contains $F$ packets $\left\{\mathbf{w}_{k, f}\right\}_{f=1}^F$,  each containing $B$ bits, i.e., $\mathbf{w}_{k, f} \in \mathbb{F}_2^B$. Additionally, it is assumed that each transmitter and receiver stores a part of the library before the receivers' demands are revealed. The cache size is $M_T F$ packets for the transmitters and $M_R F$ packets for the receivers. 
	
    The network functions in two distinct phases: the prefetching phase and the delivery phase, explained as follows:
	
	\textit{Prefetching Phase}: During this phase, 
    each transmitter $i$ designs a caching function 	$\phi_{i}$ that maps  
	$\mathcal{W}_F = \left\{\mathbf{w}_{k, f}: k \in [N], f \in [F]\right\}$
	into its cache while satisfying the cache size constraint as:
	\begin{equation}\label{memory_func_tx}
		U_{i} = \phi_{i}(\mathcal{W}_F), \quad |U_{i}| \leq M_{T} F, \quad \forall i \in [K_T].
	\end{equation}
	
	Likewise, each receiver 
    designs a caching function $\zeta_j$ that maps 	$\mathcal{W}_F$ into its cache, consistent with the cache size constraint, as follows: 
	\begin{equation}
		V_{j} = \zeta_{j}(\mathcal{W}_F), \quad |V_{j}| \leq M_{R} F, \quad \forall j \in [K_R].
	\end{equation}
	
	Caching functions are known at all nodes and the IRS. Similar to \cite{nader}, in this paper, caching occurs at the packet level. Moreover, all transmitters and receivers use uncoded prefetching, not being allowed to cache combinations of multiple packets as a single packet. Note that caching takes place regardless of the receivers' demands in the future. Additionally, content placement ensures that each packet is cached in at least one transmitter, removing the backhaul cost and latency. To achieve this, we need to have
		$K_T M_T F \geq NF$.
	
	We also consider that $M_R < N$, preventing receivers from caching the entire library, as there would be no need to formulate this problem otherwise. Moreover, we assume $N \geq K_R$, enabling each receiver to request a distinct file and allowing us to analyze the worst-case scenario.
	
	\textit{Delivery Phase}: During this phase, each receiver $j \in [K_R]$ requests an arbitrary file $W_{d_j}$ from the library. We denote the
	demand vector as $\mathbf{d} \triangleq(W_{d_1}, W_{d_2}, \ldots, W_{d_{K_R}}) \in[N]^{K_{R}}$. Then, the transmitters send the requested files. Note that, as determined during the prefetching phase, each receiver may have already cached some packets of its desired file. As a result, the transmitters need to transmit only the uncached packets of each file to the corresponding receiver.

More specifically and similar to \cite{nader}, each transmitter uses a random Gaussian coding scheme $\psi: \mathbb{F}_2^B \rightarrow \mathbb{C}^{\tilde{B}}$ with a rate of $\log P+o(\log P)$ to encode each of its cached packets into a coded packet consisting of $\tilde{B}$ complex symbols. Thus, each coded packet carries one
	DoF. The coded version of the packet 	$\mathbf{w}_{k, f}$ is denoted as 	$\tilde{\mathbf{x}}^{}_{k, f} = \tilde{\mathbf{w}}_{k, f} \triangleq \psi\left(\mathbf{w}_{k, f}\right)$. Subsequently, communication takes place over $H$ blocks, with each block comprising $\tilde{B}$ time slots. In each block $m \in [H]$, a subset of requested packets ($\mathcal{D}_m$) is transmitted interference-free to a subset of receivers ($\mathcal{R}_m$) with the aid of the active IRS.
	
	Particularly, during block $m \in [H]$, each transmitter $i$ is assumed to employ a one-shot linear scheme to transmit a linear combination of the cached coded packets in	$\mathcal{D}_m$, denoted as 	$\mathbf{x}^{[i]}(m) \in \mathbb{C}^{\tilde{B}}$, as follows:
	\begin{equation}\label{linear_comp_tx}
		\mathbf{x}^{[i]}(m) =\sum_{\substack{(k, f): \\ \mathbf{w}_{k, f} \in {U}_i \cap \mathcal{D}_m}} \tilde{\mathbf{V}}^{[i]}_{k, f}(m) \tilde{\mathbf{x}}^{}_{k, f},
	\end{equation}
	where 		$\tilde{\mathbf{V}}^{[i]}_{k, f}(m) \in \mathbb{C}^{\tilde{B} \times \tilde{B}}$ is a diagonal matrix with
	$\tilde{B}$ complex beamforming coefficient in each block 		$m \in [H]$ as its elements. Each transmitter $i$‌ can send a linear combination of the coded packets using this matrix.
	
	In each block $m \in [H]$, the received signal at each receiver $j$ from the subset 	$\mathcal{R}_m$ is given as follows:
	\begin{equation}
		\mathbf{y}^{[j]}(m)=\sum_{i=1}^{K_T} \tilde{\mathbf{H}}^{[ji]}(m) \mathbf{x}^{[i]}(m)+\mathbf{z}^{[j]}(m),
	\end{equation}
	where	$\tilde{\mathbf{H}}^{[ji]}(m) \in \mathbb{C}^{\tilde{B} \times \tilde{B}}$ is the diagonal matrix composed of the channel coefficients between transmitter $i$ and receiver $j$ and 	$\mathbf{z}^{[j]}(m) \in \mathbb{C}^{\tilde{B}}$ is the Gaussian noise. Then, each receiver $j$ forms a linear combination of its received signal and its cache contents, denoted by 
		$\mathcal{L}_{j, m}(\mathbf{y}^{[j]}(m), {V}_j)$
	to decode its desired coded packet 	$\tilde{\mathbf{w}}_{d_j, f}$. 
	
	The transmission of packets $\mathcal{D}_m$ during block $m \in [H]$ is successful if, at the receiver's side, there exists a linear combination 
    such that:
	\begin{equation}\label{channel_created}
		\mathcal{L}_{j, m}\left(\mathbf{y}^{[j]}(m), {V}_j\right) = \tilde{\mathbf{w}}_{d_j, f} + \mathbf{z}^{[j]}(m).
	\end{equation}
	
	The channel described in \eqref{channel_created} represents a point-to-point channel with a $\log P+o(\log P)$ capacity. Since each packet $\tilde{\mathbf{w}}_{d_j, f}$ is coded with the same rate, it can be decoded with a vanishing error probability as $B$ increases. 
	
	Note that the one-shot linear sum-DoF of 	$|\mathcal{D}_m|$ is achievable in each block $m \in [H]$ as each packet carries one DoF. 
    Hence, for a given caching realization, we define the one-shot linear sum-DoF as the maximum achievable one-shot linear sum-DoF for the worst-case demands, as follows:
	\begin{equation}\label{DoF_def}
		\mathrm{DoF_{sum}}(M_T, M_R)
		= \inf _{\mathbf{d}} \sup _{H,\left\{\mathcal{D}_m\right\}_{m=1}^H} \frac{\left|\bigcup_{m=1}^H \mathcal{D}_m\right|}{H}.
	\end{equation}
		
	Note that $\mathrm{DoF_{sum}}$ provides an approximate capacity of the system in a high SNR regime, which is accurate within $\operatorname{o}(\log(P))$. In the following, we explore the trade-off between the caching parameters $\left(M_T, M_R\right)$, the number of IRS elements $Q$, and the sum-DoF $\mathrm{DoF_{sum}}$.

\subsection{Preliminaries}\label{preliminaries_sec}
In the following, we highlight two results that are used throughout the paper. The first is a lemma on changing the network topology using an active IRS. The second is a combinatorial lemma that is employed in the achievability proofs.
\begin{lemalpha}\label{lemma_bafghi} Given the network described in Section \ref{systemmodelsubsection}, an active IRS with $Q$ elements, with $u$-th element  modeled by     $q^{[u]}(t) =\rho^{[u]}(t) e^{j \phi^{[u]}(t)}$ for
    $u \in [Q]$,
	$\phi^{[u]}(t) \in[0,2 \pi]$, and
	$\rho^{[u]}(t) \in \mathbb{R}^{+}$, is able to change the topology  by realizing a network matrix $\mathcal{N} = [n_{i, j}]$, which has $Q$ zero elements with probability equal to $1$.
    \end{lemalpha}
    \begin{proof}
        Please refer to \cite{bafghi_TCom, bafghi_xnetwork}. 
    \end{proof}

    \begin{lemalpha}[Baranyai's theorem]\label{Baranyai's Theorem}
Let $2 \leq k<n$ be integers such that $k$ divides $n$. There exist $\binom{n-1}{k-1}$ partitions of $[n]$ into $k$-element subsets such that no subset appears in more than one partition.
\end{lemalpha}

\begin{proof}
    The proof can be found in \cite{Baranyai_main} or  in standard combinatorics textbooks such as \cite{combinatorics_book}.
\end{proof}
	We refer to this collection of $\binom{n-1}{k-1}$ partitions as a $(\frac{n}{k}, k)$-subset partition. To illustrate Baranyai's theorem, consider the following example.
        \begin{example*}
        For $k = 2$ and $n=6$, there exist $\binom{n-1}{k-1} = \binom{5}{1}$ partitions of $[n]=[6]$ into $2$-element subsets such that no subset appears in more than one partition. One example of a $(3, 2)$-subset partition is given by:
      \begin{align*}
			\biggl\{&\Bigl\{
			\{1, 2\}, \{3, 4\}, \{5, 6\}\Bigl\},
			\Bigl\{
			\{1, 3\}, \{2, 5\}, \{4, 6\}\Bigl\}, \Bigl\{
			\{1, 4\}, \{2, 6\}, \{3, 5\}\Bigl\}, \\
			&
			\Bigl\{
			\{1, 5\},\{2, 4\}, \{3, 6\}\Bigl\},
            \Bigl\{
			\{1, 6\}, \{2, 3\}, \{4, 5\}\}\Bigl\}\biggl\}.
		\end{align*}	
    \end{example*}
    	\section{Main Results and Discussion}
           \label{Main_Results_and_Discussions}
        \color{black}
        \textcolor{black}{In this section, we present our main results on the one-shot sum-DoF of the network and discuss their implications.} To this end, we define the parameters ${\mu_T} \triangleq \frac{K_{T}M_{T}}{N}$ and ${\mu_R} \triangleq \frac{K_{R}M_{R}}{N}$
	at the transmitter and receiver sides, respectively. In the following, we first present the results for the case ${\mu_T} = 1$, and then extend them to the case $\mu_T > 1$. 
    
\begin{theorem}\label{Theorem_RIS_without_extended_symbol}
		Consider a cache-aided interference network with $K_T$ transmitters, $K_R$ receivers, ${\mu_T} = 1$, and ${\mu_R} \in [K_R-1]$,  assisted by an IRS‌ with $Q = L (L+1) $ elements, $L \in [\min\{K_T-1, K_R-1\}]$. Then, the one-shot sum-DoF of 
		$\mathrm{DoF_{sum}} = \operatorname{min}\{{\mu_R}+1+L, K_R\}$
		is achievable.
	\end{theorem}
    \begin{theorem}\label{Theorem_RIS_without_extended_symbol_tau_geq_1}
			Consider a cache-aided interference network with $K_T = M \mu_T$ transmitters, $K_R$ receivers, ${\mu_T} \in [2:K_T]$, and ${\mu_R} \in [K_R-1]$, assisted by an IRS‌ with $Q = \mu_T L(L+1) $ elements, $L \in [\min\{M-1, K_R-1\}]$. Then, the one-shot sum-DoF of 
			$\mathrm{DoF_{sum}} = \operatorname{min}\{{\mu_R}+\mu_T+L, K_R\}$
			is achievable.
		\end{theorem}
In the following, we highlight several discussions related to the results of Theorems~\ref{Theorem_RIS_without_extended_symbol} and~\ref{Theorem_RIS_without_extended_symbol_tau_geq_1}:
\begin{enumerate}
     \item (Discussion of the $K_T = M\mu_T$ constraint for Theorem \ref{Theorem_RIS_without_extended_symbol_tau_geq_1})
     This condition ensures that the transmitters can be partitioned into disjoint groups of size $\mu_T$ in each transmission block, such that no group appears in more than one block (refer to Lemma \ref{Baranyai's Theorem}). If $K_T$ is not divisible by $\mu_T$, we suggest two possible modifications to the scheme. 
     First, 
     one solution is to apply the achievability scheme corresponding to a smaller normalized cache size $\mu_T^{\prime}$, where $1 \leq \mu_T^{\prime} < \mu_T$ and $\mu_T^{\prime} \mid K_T$ by ignoring access to additional cached subfiles and modifying the prefetching and delivery phases accordingly. 
     This results in a DoF loss of $\mu_T-\mu_T^{\prime}$. However, the achieved sum-DoF, given by $\operatorname{min}\{\mu_T'+\mu_R+L, K_R\}$, can provide a significant gain over the baseline sum-DoF $\operatorname{min}\{\mu_T+\mu_R, K_R\}$\cite{nader}, depending on the number of IRS elements.
Alternatively, 
let $K_T = M \mu_T + r$ where $0 < r < \mu_T$. 
We can arbitrarily ignore $r$ transmitters and perform the prefetching and delivery phases using the remaining ${K}'_T = M\mu_T$ transmitters. To assess the resulting loss in the achievable sum-DoF, we compare this setting with a baseline employing $M'\mu_T$ transmitters where $M' = M+1$. Since $L \in \operatorname{min}\{M'-1, K_R-1\}$, the maximum loss in the sum-DoF is at most $1$.

     \item (Scaling of sum-DoF with number of elements)
     When 
     $\mathrm{DoF_{sum}} < K_R$, 
     the sum-DoF scales linearly with $L$. Since $Q = L(L+1)$ in Theorem \ref{Theorem_RIS_without_extended_symbol} and $Q = \mu_T L(L+1)$ in Theorem \ref{Theorem_RIS_without_extended_symbol_tau_geq_1}, the improvement in sum-DoF as a function of the number of elements scales as $O(\sqrt{Q})$.
     When $\mathrm{DoF_{sum}} = K_R$,  
      increasing $Q$ provides no further DoF gain, as the maximum sum-DoF is already achieved. Therefore, in cases where the maximum sum-DoF can be achieved using $Q' \leq Q$ IRS elements, the remaining elements are kept silent. Moreover, 
     if no integer $L$ satisfies $Q=L(L+1)$ for Theorem 1 or $Q=\mu_T L(L+1)$ for Theorem 2, we select the largest $Q^{\prime} \leq Q$ of that form. This implies that only $Q^{\prime}$ elements of the IRS are activated. 

    \item (Results in the presence of a passive IRS) Based on Lemma \ref{lemma_bafghi}, a key advantage of an active IRS with amplitude control is that it guarantees the realization of a desired network matrix with probability one using a finite (and small) number of elements, thereby providing a \textit{deterministic} DoF improvement. In contrast, when a passive IRS is employed,
    its reduced capabilities cause
    the set of feasible network matrices in a given time slot to depend on the realization of the channel coefficients in that slot, which introduces randomness and makes the analysis \textit{probabilistic} \cite{bafghi_TCom, bafghi_xnetwork}. 
    Due to the same channel assumptions and a similar approach for converting a fully connected interference network into a desired partially connected one as in \cite{bafghi_TCom, bafghi_xnetwork}, our results are expected to hold in a probabilistic sense for a passive IRS. Specifically, under mild conditions on the channel coefficients (e.g., \cite[equations (25)–(29)]{bafghi_TCom} or \cite[Theorem 10]{bafghi_xnetwork}), satisfied by common fading models such as Rayleigh fading, the sum-DoF is expected to converge in probability to that of the active case as $Q \rightarrow \infty$, with a rate of at least $O(\frac{1}{Q})$. A detailed and rigorous analysis for the passive case is left for future work.

\item (Active IRS power constraint)
        Regarding the assumptions on the IRS amplification, we followed the approach in \cite{bafghi_TCom, bafghi_xnetwork}. Accordingly, we have not explicitly imposed a power constraint $\rho_{\text{max}}$ on the elements of active IRS. This choice is mainly motivated by the DoF analysis being conducted in the high SNR regime, where the transmit power tends to infinity in the DoF definition and, consequently, such a power constraint  may become irrelevant. More precisely, if a power constraint is imposed, e.g., $\left|q^{[u]}(t)\right| = \rho^{[u]}(t)< \rho_{\text{max}}, \forall u \in [Q]$, then the results remain applicable over the fraction of time slots in which the constraint is satisfied, i.e., when $\operatorname{Pr}\left\{\left|q^{[u]}(t)\right| = \rho^{[u]}(t) <\rho_{\text{max}}, \forall u \in [Q]\right\}$ holds. By choosing $\rho_{\text{max}}$ sufficiently large, this probability approaches one. Moreover, due to the imposed bound on the IRS amplification power, only a constant noise term is introduced, which can be neglected in the DoF analysis as the power of signal goes to infinity, i.e., $P \rightarrow \infty$, independently of the IRS power constraint.
    \item (Upper bound)\label{Upper_bound}
Our results can be upper bounded by $\operatorname{min}\{K_T Q+\mu_R, K_R\}$. This follows by considering a genie-aided channel in which all transmitters have access to the entire content library and can fully cooperate in determining their channel inputs, and the IRS is assumed to be more powerful, capable of individually amplifying and forwarding the signals from each transmitter, rather than forwarding the summation of amplified signals. This genie-aided model corresponds to a MISO broadcast channel with a single transmitter equipped with $K_T Q$ antennas and $K_R$ single-antenna receivers, each with a normalized receiver-side cache parameter $\mu_R$. Since providing additional side information and enhanced capabilities cannot reduce the capacity region, the capacity region of our original channel in both Theorems is contained within that of the genie-aided channel. 
Under uncoded cache placement and one-shot linear delivery, the optimal sum-DoF of the genie-aided channel is $\operatorname{min}\{K_T Q+\mu_R, K_R\}$  \cite{Petros_MIMO}, which therefore constitutes an upper bound for our system.
\end{enumerate}
			
   \color{black}

	\section{Achievability Scheme}
 \label{achivability_Schemes:sec}

     \color{black} 
    In the following, we propose the achievability scheme for the integer\footnote{For the non-integer values, we can use memory-sharing, similar to time-sharing \cite{nader, main_caching_maddah_niesen}.} values of ${\mu_T} \in [K_T]$ and ${\mu_R} \in [K_R-1]$. \textcolor{black}{To this end, since Theorems~\ref{Theorem_RIS_without_extended_symbol} and~\ref{Theorem_RIS_without_extended_symbol_tau_geq_1} share common structure in their prefetching and delivery phases, we first present these common aspects and then explain the rest of the scheme for each theorem separately.} 
	
	\textit{Prefetching Phase}: In this phase, similar to \cite{nader}, each file $W_k$ is partitioned into ${K_T \choose {\mu_T}}{K_R \choose {\mu_R}}$ distinct subfiles. Thus, each file $W_k$ can be represented as follows:
	\begin{equation}
        W_{k}=\left\{W_{k, \mathcal{S}, \mathcal{R}}: \mathcal{S} \subseteq\left[K_{T}\right],|\mathcal{S}|={\mu_T}, \mathcal{R} \subseteq\left[K_{R}\right],|\mathcal{R}|={\mu_R}\right\}.
	\end{equation}
	
	Based on the above partitioning, each subfile $W_{k, \mathcal{S}, \mathcal{R}}$ is cached at transmitter $i$, if $i \in \mathcal{S}$. Therefore, each transmitter stores $N {K_T -1 \choose {\mu_T}-1}{K_R \choose {\mu_R}}$ subfiles in its cache memory, which is consistent with the size constraint.
	Similarly, each subfile $W_{k, \mathcal{S}, \mathcal{R}}$ is cached at  receiver $j$ if $j \in \mathcal{R}$. Likewise, each receiver caches $N {K_T \choose {\mu_T} } {K_R-1 \choose {\mu_R}-1}$ subfiles, satisfying the cache size constraint.
	
	\textit{Delivery Phase}: In this phase, based on the users' demand vector $\mathbf{d} = (W_{d_1}, W_{d_2}, \ldots, W_{d_{K_R}})$, the transmitters are required to deliver the following $K_{R} {K_{T} \choose {\mu_T}} {K_{R}-1 \choose {\mu_R}}
	$ subfiles:
	\begin{align}\label{send_general_tau1}
		\left\{
        W_{d_{j}, \mathcal{S}, \mathcal{R}}: j \in\left[K_{R}\right], \mathcal{S} \subseteq\left[K_{T}\right], |\mathcal{S}|={\mu_T}, \mathcal{R} \subseteq\left[K_{R}\right]\backslash \{j\},
         |\mathcal{R}|={\mu_R}
         \right\}.
	\end{align}
	 
	
	
\textcolor{black}{In the following, we present the proof sketch of each theorem.}
    \subsection{Proof of Theorem \ref{Theorem_RIS_without_extended_symbol}}
    \color{black}
	\begin{sproof}[Sketch of proof] 
        We consider two cases: (I) $\mu_R+L+1 \geq K_R$ and (II) $\mu_R+L+1<K_R$. In case (I), we demonstrate the achievability of the maximum possible sum-DoF. Specifically, we first illustrate how the set of requested subfiles can be partitioned into disjoint subsets of size $K_R$, each \textcolor{black}{corresponding to} one receiver, \textcolor{black}{as shown in Lemma \ref{main_lemma_theorem_mut1_oneshot}.} We then explain how the $K_R$ subfiles \textcolor{black}{in each set} can be delivered interference-free to \textcolor{black}{their} intended receivers \textcolor{black}{within a} communication block. \textcolor{black}{This is achieved by realizing a desired partially-connected network through appropriate IRS configuration and beamforming design, as described in Lemmas \ref{lemma_bafghi} and \ref{main_lemma_theorem_mut1_oneshot_2}}.
        Since there is no overlap between the caches of transmitters in this case, the beamforming coefficients \textcolor{black}{are restricted to binary values, i.e., $0$ or $1$, indicating whether a subfile is selected for transmission.} In case (II), the same procedure is followed, but only for a subset of receivers. \textcolor{black}{As a result,} in each communication block, \textcolor{black}{only the selected} receivers achieve a DoF of $1$, while \textcolor{black}{the remaining receivers achieve} a DoF‌ of $0$. The complete proof is provided in Appendix A. 
	\end{sproof}

    \begin{figure}[htbp]
\centering

\begin{minipage}{0.32\linewidth}
    \centering
    \includegraphics[width=\linewidth]{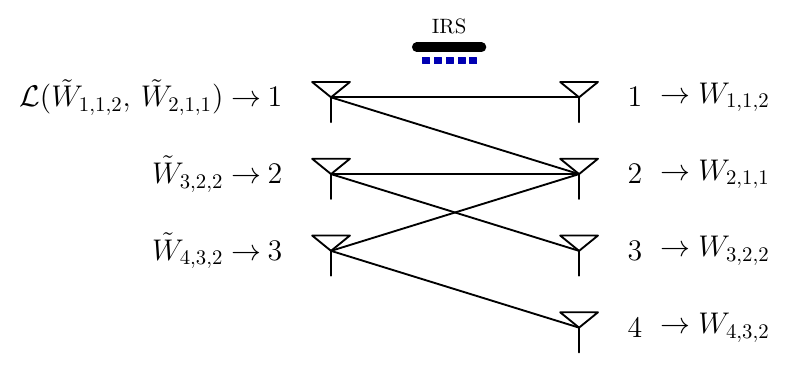}
    \footnotesize(a)
\end{minipage}
\begin{minipage}{0.32\linewidth}
    \centering
    \includegraphics[width=\linewidth]{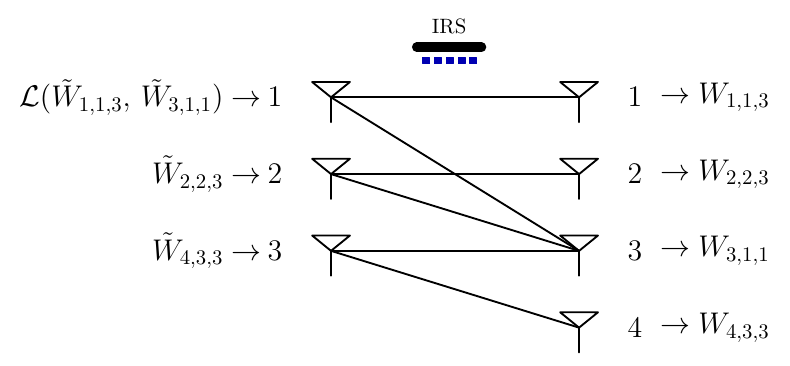}
   \footnotesize (b)
\end{minipage}
\begin{minipage}{0.32\linewidth}
    \centering
    \includegraphics[width=\linewidth]{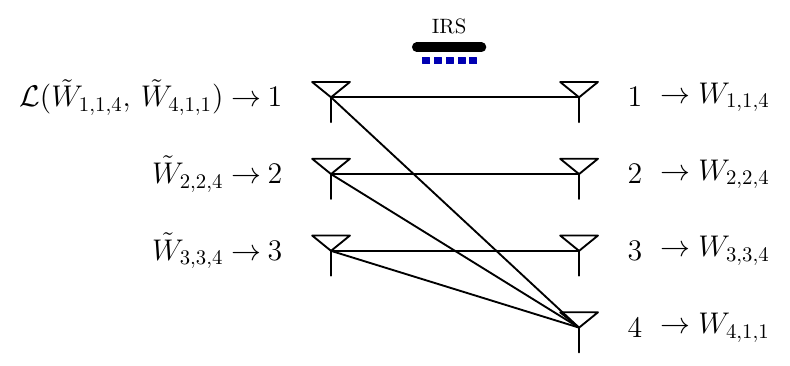}
   \footnotesize (c)
\end{minipage}

\vspace{0.3em}

\begin{minipage}{0.32\linewidth}
    \centering
    \includegraphics[width=\linewidth]{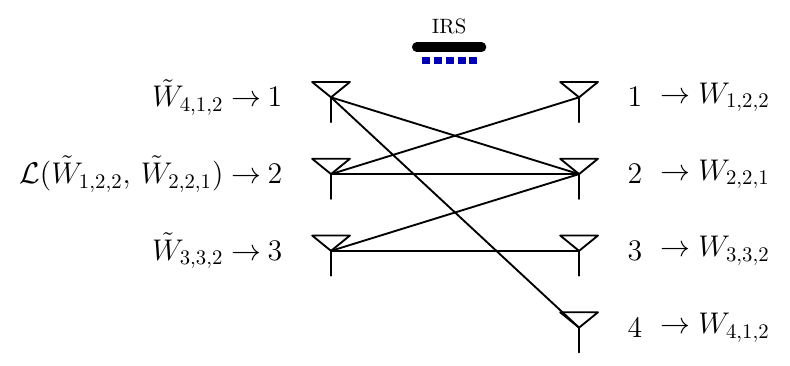}
   \footnotesize (d)
\end{minipage}
\begin{minipage}{0.32\linewidth}
    \centering
    \includegraphics[width=\linewidth]{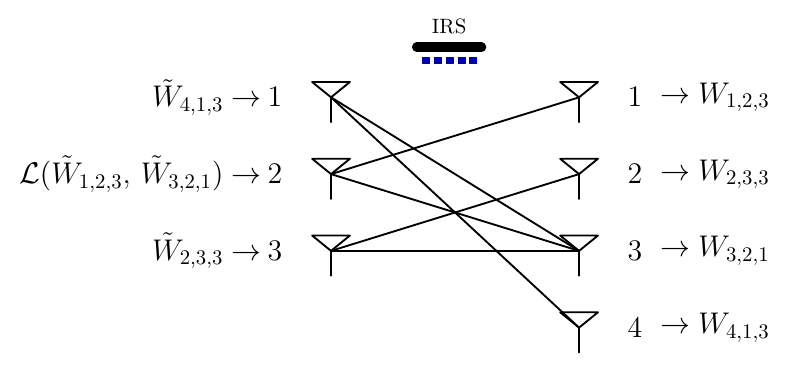}
   \footnotesize (e)
\end{minipage}
\begin{minipage}{0.32\linewidth}
    \centering
    \includegraphics[width=\linewidth]{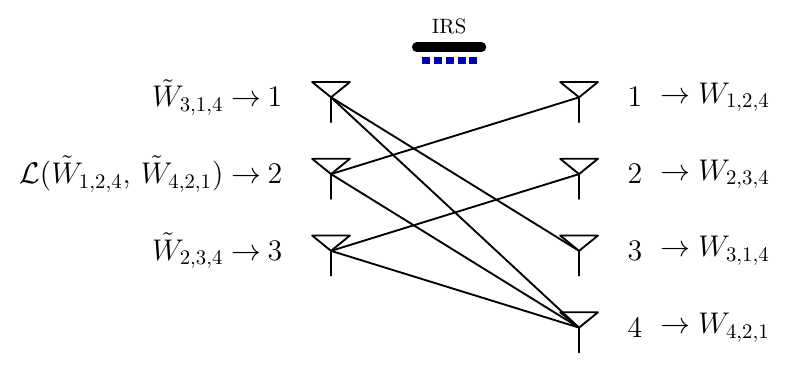}
    \footnotesize(f)
\end{minipage}

\vspace{0.3em}

\begin{minipage}{0.32\linewidth}
    \centering
    \includegraphics[width=\linewidth]{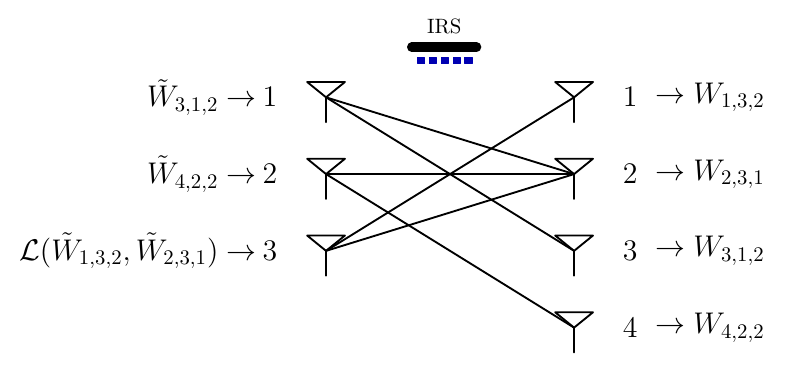}
    {\footnotesize(g)}
\end{minipage}
\begin{minipage}{0.32\linewidth}
    \centering
    \includegraphics[width=\linewidth]{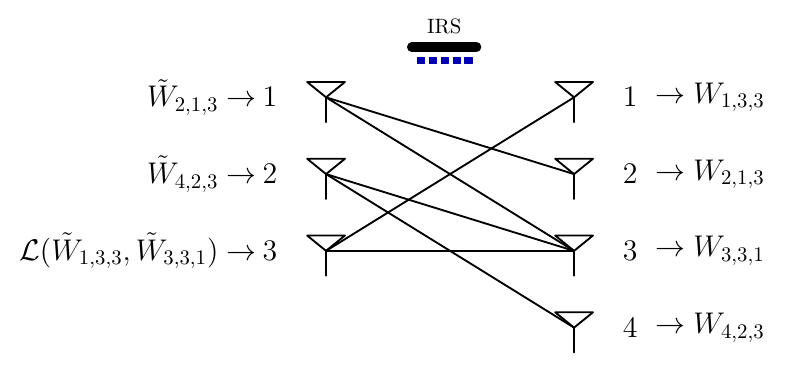}
 \footnotesize   (h)
\end{minipage}
\begin{minipage}{0.32\linewidth}
    \centering
    \includegraphics[width=\linewidth]{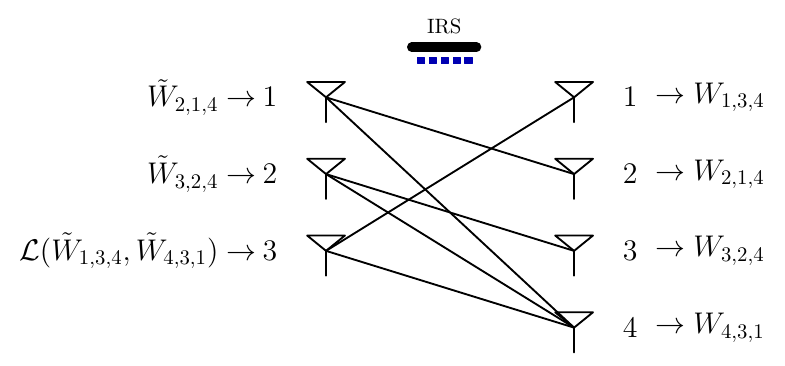}
    \footnotesize(i)
\end{minipage}

\caption{\color{black} Illustration of transmission in a cache-aided interference channel assisted by a $6$-element active IRS with
						$K_T = 3$,
						$K_R = 4$,
						${\mu_T} = 1$,
						$\mu_R = 1$. At each communication block, the subfiles transmitted by the transmitters are shown on the left, while the subfiles that each receiver can decode are shown on the right.}
\label{example_Kt3_Kr4}
\end{figure}
\textcolor{black}{In the following, we explain the achievability scheme for case (I) of Theorem \ref{Theorem_RIS_without_extended_symbol} through an example.}
			\begin{example}
                Consider a cache-aided communication system assisted by an active IRS with $Q = 6$, $K_T = 3$, and $K_R = 4$, where each transmitter and receiver is equipped with cache memories of sizes $M_T = 4$ and $M_R = 3$‌ files, respectively. If the library consists of $N = 12$ files, denoted by 	$\mathcal{W} = \left\{W_1,\ldots, W_{12}\right\}$, the 
                {maximum} DoF of $1$ is achievable for each receiver.  
                \normalfont
                
              Note that ${\mu_T} = \frac{K_T M_T}{N} =1$ and 	${\mu_R} = \frac{K_R M_R}{N} = 1$.  The prefetching and delivery phases are described as follows: 
              
                \textit{Prefetching Phase:} 
                In this phase, each file $W_k, k \in [12]$ is divided into  ${3 \choose 1}{4 \choose 1} = 12$  subfiles $W_{k, \mathcal{S}, \mathcal{R}}$, where $\mathcal{S} \subseteq [K_T] = [3]$ and $\mathcal{R} \subseteq [K_R] = [4]$, such that $|\mathcal{S}|={1}$ and $|\mathcal{R}|={1}$.
                For instance, file 
				$W_1$ is divided as follows: 
				\begin{align*}
					\{
                    W_{1, 1, 1}, W_{1, 1, 2}, W_{1, 1, 3}, W_{1, 1, 4}, 
					W_{1, 2, 1}, W_{1, 2, 2}, 
                    W_{1, 2, 3}, W_{1, 2, 4}, 
					W_{1, 3, 1}, W_{1, 3, 2}, W_{1, 3, 3}, W_{1, 3, 4}
					\},
				\end{align*}
				and each of the remaining files is divided into $12$ subfiles in the same manner. Then, $W_{k, \mathcal{S}, \mathcal{R}}$ is stored in the cache of the single transmitter in $\mathcal{S}$ and the single receiver in $\mathcal{R}$.
                
                \textit{Delivery Phase}: Without loss of generality, we assume that each receiver $j$‌ requests file $W_j$‌ for all $j \in [4]$. Each receiver has already cached $3$ subfiles of its requested file. Therefore, it suffices for the transmitters to deliver the remaining $9$ subfiles per each requested file, resulting in a total of $36$ subfiles to be transmitted. 
                For instance, the following subfiles must be delivered to receiver 1:
$$
					\left\{W_{1, 1, 2}, W_{1, 1, 3}, W_{1, 1, 4},
					W_{1, 2, 2}, W_{1, 2, 3}, W_{1, 2, 4},
					W_{1, 3, 2}, W_{1, 3, 3}, W_{1, 3, 4}\right\}, 
$$ 

‌Based on Lemmas \ref{main_lemma_theorem_mut1_oneshot} and \ref{main_lemma_theorem_mut1_oneshot_2}, these $36$ subfiles can be partitioned into $9$ disjoint sets, each containing $4$ subfiles, such that the subfiles in each set can be delivered simultaneously in an interference-free manner using an appropriate IRS configuration. This is shown in Fig. \ref{example_Kt3_Kr4}.
For example in Fig. \ref{example_Kt3_Kr4}-(a), the transmit signals of transmitters $1$, $2$, and $3$ 
are given as follows:
\begin{align*}
    	\mathbf{x}^{[1]} = \tilde{W}_{1,1,2}+\tilde{W}_{2,1,1},\quad
        \mathbf{x}^{[2]} = \tilde{W}_{3, 2, 2},\quad
       \mathbf{x}^{[3]}  = \tilde{W}_{4, 3, 2},
\end{align*}
where for any $W_{j, \mathcal{S}, \mathcal{R}}$, $\tilde{W}_{j, \mathcal{S}, \mathcal{R}}$ denotes its coded version. 
The received signals are written as follows:
\begin{align*}
    	\mathbf{y}^{[1]} &= 
\tilde{\mathbf{H}}^{[11]}\tilde{W}_{1,1,2}+\tilde{\mathbf{H}}^{[11]}\underbrace{\textcolor{black}{\tilde{W}_{2,1,1}}}_{\text{cached}}+\mathbf{z}^{[1]},\\
        \mathbf{y}^{[2]}&= \tilde{\mathbf{H}}^{[21]}\underbrace{\textcolor{black}{\tilde{W}_{1,1,2}}}_{\text{cached}}+
        \tilde{\mathbf{H}}^{[21]}\tilde{W}_{2,1,1}+
        \tilde{\mathbf{H}}^{[22]}\underbrace{\textcolor{black}{\tilde{W}_{3,2,2}}}_{\text{cached}}+
        \tilde{\mathbf{H}}^{[23]}\underbrace{\textcolor{black}{\tilde{W}_{4,3,2}}}_{\text{cached}}+\mathbf{z}^{[2]},\\
       \mathbf{y}^{[3]} &= \tilde{\mathbf{H}}^{[32]}\tilde{W}_{3,2,2}+\mathbf{z}^{[3]},\\
       \mathbf{y}^{[4]} &= \tilde{\mathbf{H}}^{[43]}\tilde{W}_{4,3,2}+\mathbf{z}^{[4]},
\end{align*}

As seen by the received signals, each receiver can cancel the subfiles marked by bracket using its cached content. In addition, several interference terms are eliminated by removing undesired cross-links using the active IRS and therefore do not appear in the received signals. In this example, $Q = L(L+1) = 2 \times 3 = 6$ elements removes $L = 2$ undesired links for each of the $L+1 = 3$ transmitters. The same procedure applies to the remaining communication blocks, resulting in an achievable per-user DoF of $1$.

\color{black}
\end{example}

\color{black}
\textcolor{black}{We now discuss the proof of Theorem~\ref{Theorem_RIS_without_extended_symbol_tau_geq_1}. Since each subfile is available at $\mu_T > 1$ transmitters, cooperation can be exploited to enable zero-forcing (ZF) of the outgoing interference, thereby improving the achievable sum-DoF.} 
            \subsection{Proof of Theorem \ref{Theorem_RIS_without_extended_symbol_tau_geq_1}}


		
		\begin{sproof}[Sketch of proof] 
We consider two cases: (I) $\mu_R+\mu_T+L \geq K_R$ and (II) 	$\mu_R+\mu_T+L < K_R$. In case (I), we present the achievability of the maximum possible sum-DoF. To this end, we further divide each subfile into smaller subfiles \textcolor{black}{in order to exploit transmitter cooperation}. \textcolor{black}{We first show how a set of $\mu_T+\mu_R$ subfiles of a particular form can be decoded interference-free in a network of size $\mu_T \times (\mu_T+\mu_R)$, as demonstrated in Lemmas \ref{lemma_1RX_without_interference} and \ref{lemma_tplustau_RX_without_interference}. Building upon these lemmas}, we then explain how $K_R$‌ distinct smaller subfiles in each communication block, together with appropriate IRS configurations, can be selected so that each subfile at the intended receiver is decoded interference-free, \textcolor{black}{as detailed in Lemmas \ref{lemma1_tau_geq1} and \ref{lemma2_tau_geq1}}. Since $\mu_T > 1$, the transmitter caches overlap, which enables ZF of interference through transmitter cooperation. To benefit from cooperation, we select $M$ distinct subsets of transmitters in each communication block, each of size $\mu_T$, ensuring that these subsets differ across blocks. \textcolor{black}{The existence of such a partitioning is guaranteed by Lemma \ref{Baranyai's Theorem}}. In case (II), the same procedure is applied, but only to a subset of receivers. Consequently, in each communication block, only the selected receivers achieve a DoF of $1$, while the remaining receivers achieve a DoF of $0$. The complete proof is provided in Appendix B. 
		\end{sproof}
        
        In the following, an example is provided to illustrate the proof of Theorem \ref{Theorem_RIS_without_extended_symbol_tau_geq_1} step-by-step.
        			\begin{example}
                    \label{Example_th2}
				Consider a cache-aided communication system assisted by an active IRS with $Q = 4$, $K_T = 4$, and $K_R = 4$, where each transmitter and receiver is equipped with cache memories of sizes $M_T = 2$ and $M_R = 1$‌ files, respectively. If the library consists of $N = 4$ files, denoted by 	$\mathcal{W} = \left\{W_1, W_2, W_3, W_4\right\}$, the maximum DoF of $1$ is achievable for each receiver.  
\normalfont

                      Note that ${\mu_T} = \frac{K_T M_T}{N} =2$ and 	${\mu_R} = \frac{K_R M_R}{N} = 1$.  The prefetching and delivery phases are described as follows: 
              
                \textit{Prefetching Phase:} 
                In this phase, each file $W_k, k \in [4]$ is divided into  ${4 \choose 2}{4 \choose 1} = 24$  disjoint subfiles $W_{k, \mathcal{S}, \mathcal{R}}$ where $\mathcal{S} \subseteq [K_T] = [4]$ and $\mathcal{R} \subseteq [K_R] = [4]$, such that $|\mathcal{S}|={2}$ and $|\mathcal{R}|={1}$.
                For instance, file 
				$W_1$ is divided as follows: 
				\begin{align*}
					\{
					&
                    W_{1, 12, 1}, W_{1, 13, 1}, W_{1, 14, 1}, W_{1, 23, 1}, W_{1, 24, 1}, W_{1, 34, 1},\\
                    &
                    W_{1, 12, 2}, W_{1, 13, 2}, W_{1, 14, 2}, W_{1, 23, 2}, W_{1, 24, 2}, W_{1, 34, 2},
                    \\
                    & W_{1, 12, 3}, W_{1, 13, 3}, W_{1, 14, 3}, W_{1, 23, 3}, W_{1, 24, 3}, W_{1, 34, 3},\\
                    &
                    W_{1, 12, 4}, W_{1, 13, 4}, W_{1, 14, 4}, W_{1, 23, 4}, W_{1, 24, 4}, W_{1, 34, 4} 
                    \},
				\end{align*}
				and each of the remaining files is divided into $24$ subfiles in the same manner. Then, $W_{k, \mathcal{S}, \mathcal{R}}$ is stored in the caches of the two transmitters in $\mathcal{S}$ and the single receiver in $\mathcal{R}$.
                
                \textit{Delivery Phase}: Without loss of generality, we assume that receiver $j$‌ requests file $W_j$‌ for all $j \in [4]$. Each receiver has already cached $6$ subfiles of its requested file. Therefore, it suffices for the transmitters to deliver the remaining $18$ subfiles per file, resulting in a total of $72$ subfiles to be transmitted. For example, the subfiles that must be delivered to receiver $1$ are given by
				\begin{align*}
					\{
			          &
                    W_{1, 12, 2}, W_{1, 13, 2}, W_{1, 14, 2}, W_{1, 23, 2}, W_{1, 24, 2}, W_{1, 34, 2},
                    \\
                    & W_{1, 12, 3}, W_{1, 13, 3}, W_{1, 14, 3}, W_{1, 23, 3}, W_{1, 24, 3}, W_{1, 34, 3},\\
                    &
                    W_{1, 12, 4}, W_{1, 13, 4}, W_{1, 14, 4}, W_{1, 23, 4}, W_{1, 24, 4}, W_{1, 34, 4} 
                    \}.
				\end{align*}

To exploit transmitter cooperation, each subfile $W_{j, \mathcal{S}, \mathcal{R}}$ is further divided  into ${{K_R-{\mu_R}-1} \choose {{\mu_T}-1}} = {{4-1-1} \choose {2-1}} = 2$ smaller subfiles $W_{j, \mathcal{S}, \mathcal{R}, \mathcal{T}}$, where 
$\mathcal{T} \subseteq [K_R]\backslash\{j\}\cup{\mathcal{R}} = [4]\backslash\{j\}\cup{\mathcal{R}}$ and $|\mathcal{T}|=\mu_T-1=1$. 
This construction enables interference cancellation via transmitter cooperation. For instance, $W_{1, 12, 2}$ 
is further divided into $W_{1, 12, 2, 3}$ and $W_{1, 12, 2, 4}$. 
\begin{figure}
    \centering
    \includegraphics[width=1\linewidth]{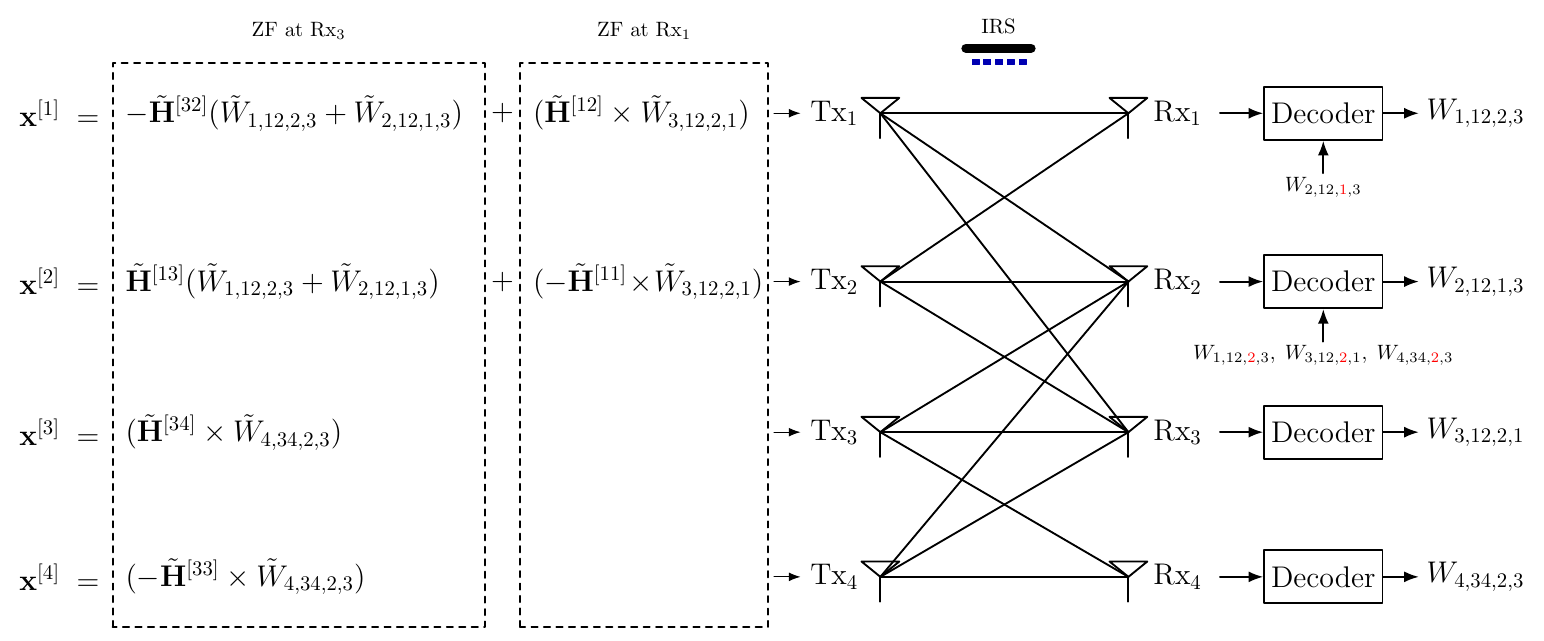}
    \caption{\textcolor{black}{\footnotesize Illustration of the linear encoding/decoding and the partially connected network realized by the $4$-element active IRS in a communication block of Example \ref{Example_th2}.}} \label{Fig_ex_th2}
\end{figure}  
Based on Lemma \ref{lemma1_tau_geq1}, these subfiles can be partitioned into $H = {{M{\mu_T} \choose {{\mu_T}}}{{K_R-1} \choose {{\mu_R}}}{{K_R-{\mu_R}-1} \choose {{\mu_T} - 1}} } = {{4 \choose {{2}}}{{3} \choose {{1}}}{{2} \choose {1}}} = 36$ disjoint sets, each containing $K_R = 4$ subfiles, such that the subfiles in each set can be delivered simultaneously in an interference-free manner using an appropriate IRS configuration. One such set is given by: 
$$
\{
W_{1, 12, 2, 3},
W_{2, 12, 1, 3},
W_{3, 12, 2, 1},
W_{4, 34, 2, 3}
\}.
$$
By Lemma \ref{lemma2_tau_geq1}, there exists a choice of beamforming coefficients together with a suitable IRS-induced network realization that enables interference-free delivery of the subfiles in each set, as illustrated in Fig. \ref{Fig_ex_th2}. The transmit signals are given by
\begin{align*}
    	\mathbf{x}^{[1]} &=
		-\tilde{\mathbf{H}}^{[32]}(\tilde{W}_{1, 12, 2, 3} + \tilde{W}_{2, 12, 1, 3})+(\tilde{\mathbf{H}}^{[12]} \times \tilde{W}_{3, 12, 2, 1}),\\
        \mathbf{x}^{[2]} &= \tilde{\mathbf{H}}^{[13]}(\tilde{W}_{1, 12, 2, 3} + \tilde{W}_{2, 12, 1, 3})+(-\tilde{\mathbf{H}}^{[11]} \times \tilde{W}_{3, 12, 2, 1}),\\
       \mathbf{x}^{[3]} &= (\tilde{\mathbf{H}}^{[34]} \times \tilde{W}_{4, 34, 2, 3}),\\
       \mathbf{x}^{[4]} &=
       (-\tilde{\mathbf{H}}^{[33]} \times \tilde{W}_{4, 34, 2, 3}),
\end{align*}
where for any $W_{j, \mathcal{S}, \mathcal{R}, \mathcal{T}}$, $\tilde{W}_{j, \mathcal{S}, \mathcal{R}, \mathcal{T}}$ denotes its coded version. 
The received signals are written as follows:
\begin{align*}
    	\mathbf{y}^{[1]} &= 
        (\tilde{\mathbf{H}}^{[13]}\tilde{\mathbf{H}}^{[12]}-\tilde{\mathbf{H}}^{[32]}\tilde{\mathbf{H}}^{[11]})\tilde{W}_{1, 12, 2, 3}+(\tilde{\mathbf{H}}^{[13]}\tilde{\mathbf{H}}^{[12]}-\tilde{\mathbf{H}}^{[32]}\tilde{\mathbf{H}}^{[11]})\underbrace{\textcolor{black}{\tilde{W}_{2, 12, 1, 3}}}_{\text{cached}}+\mathbf{z}^{[1]},\\
        \mathbf{y}^{[2]}&= (\tilde{\mathbf{H}}^{[22]}\tilde{\mathbf{H}}^{[13]}-\tilde{\mathbf{H}}^{[32]}\tilde{\mathbf{H}}^{[21]})\underbrace{\textcolor{black}{\tilde{W}_{1, 12, 2, 3}}}_{\text{cached}}+(\tilde{\mathbf{H}}^{[22]}\tilde{\mathbf{H}}^{[13]}-\tilde{\mathbf{H}}^{[32]}\tilde{\mathbf{H}}^{[21]})\tilde{W}_{2, 12, 1, 3}\\&+(\tilde{\mathbf{H}}^{[12]}\tilde{\mathbf{H}}^{[21]}-\tilde{\mathbf{H}}^{[11]}\tilde{\mathbf{H}}^{[22]})\underbrace{\textcolor{black}{\tilde{W}_{3,12,2,1}}}_{\text{cached}}+(\tilde{\mathbf{H}}^{[34]}\tilde{\mathbf{H}}^{[23]}-\tilde{\mathbf{H}}^{[33]}\tilde{\mathbf{H}}^{[24]})\underbrace{\textcolor{black}{\tilde{W}_{4,34,2,3}}}_{\text{cached}}+\mathbf{z}^{[2]},\\
       \mathbf{y}^{[3]} &= (\tilde{\mathbf{H}}^{[12]}\tilde{\mathbf{H}}^{[31]}-\tilde{\mathbf{H}}^{[11]}\tilde{\mathbf{H}}^{[32]})\tilde{W}_{3,12,2,1}+\mathbf{z}^{[3]},\\
       \mathbf{y}^{[4]} &=(\tilde{\mathbf{H}}^{[43]}\tilde{\mathbf{H}}^{[34]}-\tilde{\mathbf{H}}^{[44]}\tilde{\mathbf{H}}^{[33]})\tilde{W}_{4,34,2,3}+\mathbf{z}^{[4]},
\end{align*}

As shown by the received signals, each receiver can cancel the subfiles marked by bracket due to its cache availability. In addition, several interference terms are eliminated via ZF and by removing undesired cross-links using the active IRS and therefore do not appear in the received signals. In this example, $Q = \mu_T L(L+1) = 2 \times 1 \times 2 = 4$ elements removes $L = 1$ undesired links for each of the $\mu_T(L+1) = 4$ transmitters. The same procedure applies to the remaining communication blocks, resulting in an achievable per-user DoF of $1$. 
        \end{example}
 
        \color{black}   

 \color{black}
\section{Numerical Results}\label{sec4Ssim}
		This section provides numerical results to highlight the achievability scheme obtained in this paper. To this end, we compare our results in Theorems \ref{Theorem_RIS_without_extended_symbol} and \ref{Theorem_RIS_without_extended_symbol_tau_geq_1} with those presented in \cite{nader, Tao_Fundamental}. \textcolor{black}{These works are selected as baselines as they represent two fundamentally different state-of-the-art approaches to cache-aided interference management without IRSs. Specifically, the proposed scheme builds on the one-shot linear caching framework in \cite{nader}, enhancing its achievable performance via an active IRS, whereas \cite{Tao_Fundamental} adopts a symbol-extension-based approach that achieves tighter DoF characterizations.} We summarize the DoF per user of these works as follows:
{
\begin{align}
	 \mathrm{DoF_{}}^{\text{\cite{nader}}} &= \frac{1}{K_R}\operatorname{min}\{{\mu_T}+{\mu_R}, K_R\},\label{maddah}\\
	\mathrm{DoF_{}}^{\text{\cite{Tao_Fundamental}}} &=\begin{cases}
		1,\quad  &{\mu_R}+{\mu_T} \geq K_R, \\
		\frac{{K_R-1 \choose {\mu_R}}{K_T-1 \choose {\mu_T}}{\mu_T}}{{K_R-1 \choose {\mu_R}}{K_T-1 \choose {\mu_T}}{\mu_T}+1},\quad &{\mu_R}+{\mu_T} = K_R-1, \\
		\max\{d_1, \frac{{\mu_R}+{\mu_T}}{K_R}\},\quad &{\mu_R}+{\mu_T} \leq K_R-2,\label{Tao_DoF_equation}
	\end{cases}
\end{align}
}
where in \eqref{Tao_DoF_equation}: 
\begin{align}
d_1 = \max_{{\mu_T}^{\prime} \in [1:{\mu_T}]}
\left\{\frac{{K_R-1 \choose {\mu_R}}{K_T-1 \choose {\mu_T}^{\prime}}{K_R-{\mu_R}-1 \choose {{\mu_T}^{\prime}-1}}{{\mu_T}}^{\prime}}{{K_R-1 \choose {\mu_R}}{K_T-1 \choose {\mu_T}^{\prime}}{K_R-{\mu_R}-1 \choose {{\mu_T}^{\prime}-1}}{{\mu_T}}^{\prime}+{K_R-1 \choose {\mu_R}+1}{K_R-{\mu_R}-2 \choose {\mu_T}^{\prime}-1}{K_T \choose {{\mu_T}^{\prime}-1}}}\right\}.
\end{align}

 \begin{remark}\label{one_shot_vs_symbol_extension}
 One-shot schemes perform precoding over a single time–frequency slot (i.e., symbol-by-symbol), enabling simple implementations. In contrast, schemes based on symbol extension require signaling over a large number of time–frequency slots, with the block length typically tending to infinity to exploit the gains of asymptotic interference alignment. Such schemes are often impractical due to large delays, high-dimensional precoders/decoders, and requiring stable channel conditions over the entire extended block \cite{elgamal_oneshot_vs_extension}. In this regard, a key contribution of this work is achieving improved sum-DoF using one-shot transmission.
\end{remark}
\color{black}
\begin{figure}
    \centering
    \includegraphics[width=.7\linewidth]{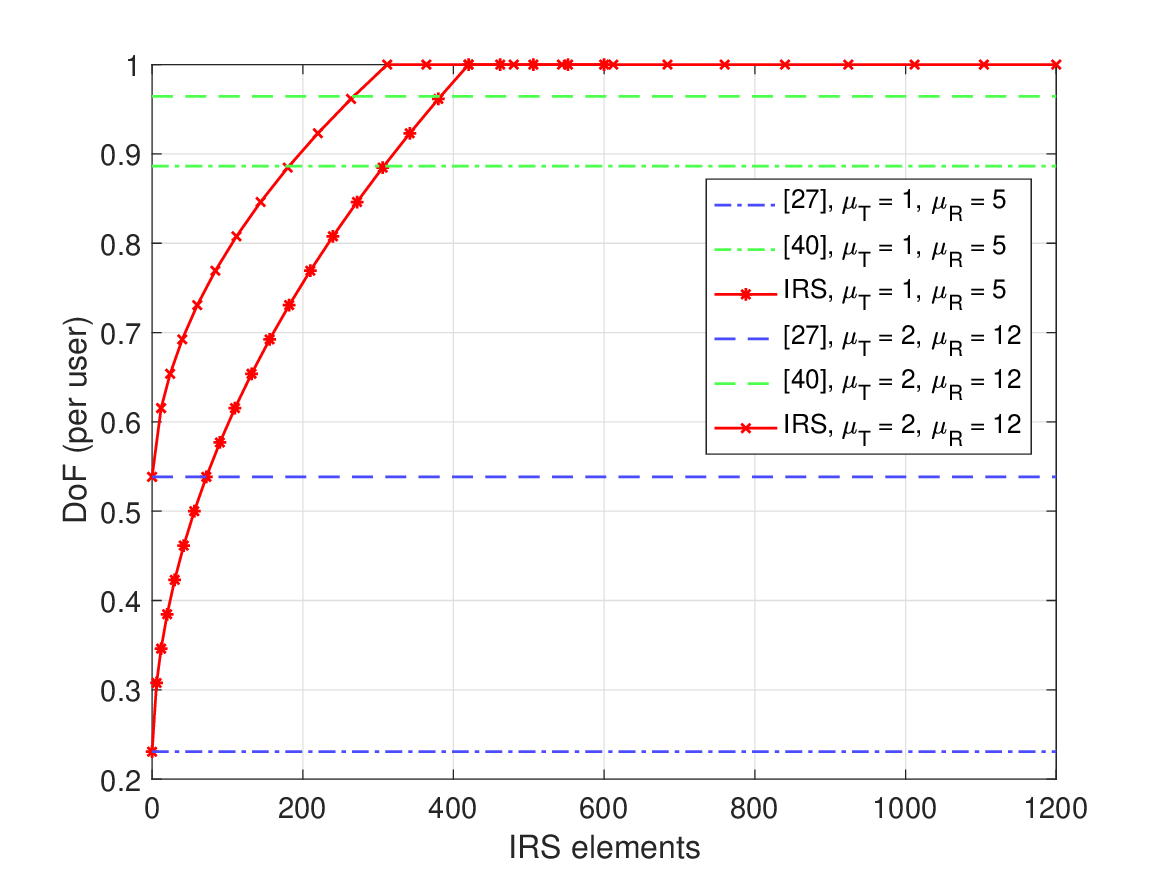}
    \caption{Achievable DoF per user versus the number of IRS elements for a cache-aided interference channel with
			$K_T = K_R = 26$.
  }
    \label{fig_new_1}
\end{figure}


In Fig. \ref{fig_new_1}, we plot the achievable DoF per user for a cache-aided interference channel with parameters $K_T = K_R = 26$. It can be observed that our proposed one-shot approach always outperforms \cite{nader}, which is nearly optimal without employing an IRS in a one-shot scenario. Furthermore, our proposed one-shot scheme demonstrates superior performance compared to \cite{Tao_Fundamental}, particularly when an IRS with a sufficient number of elements is available. 
Moreover, note that higher values of $\mu_T$ and $\mu_R$ result in a higher DoF per user. Additionally, this leads to a lower requirement for the number of IRS elements to achieve a specific DoF.

\begin{figure}
    \centering
    \includegraphics[width=.7\linewidth]{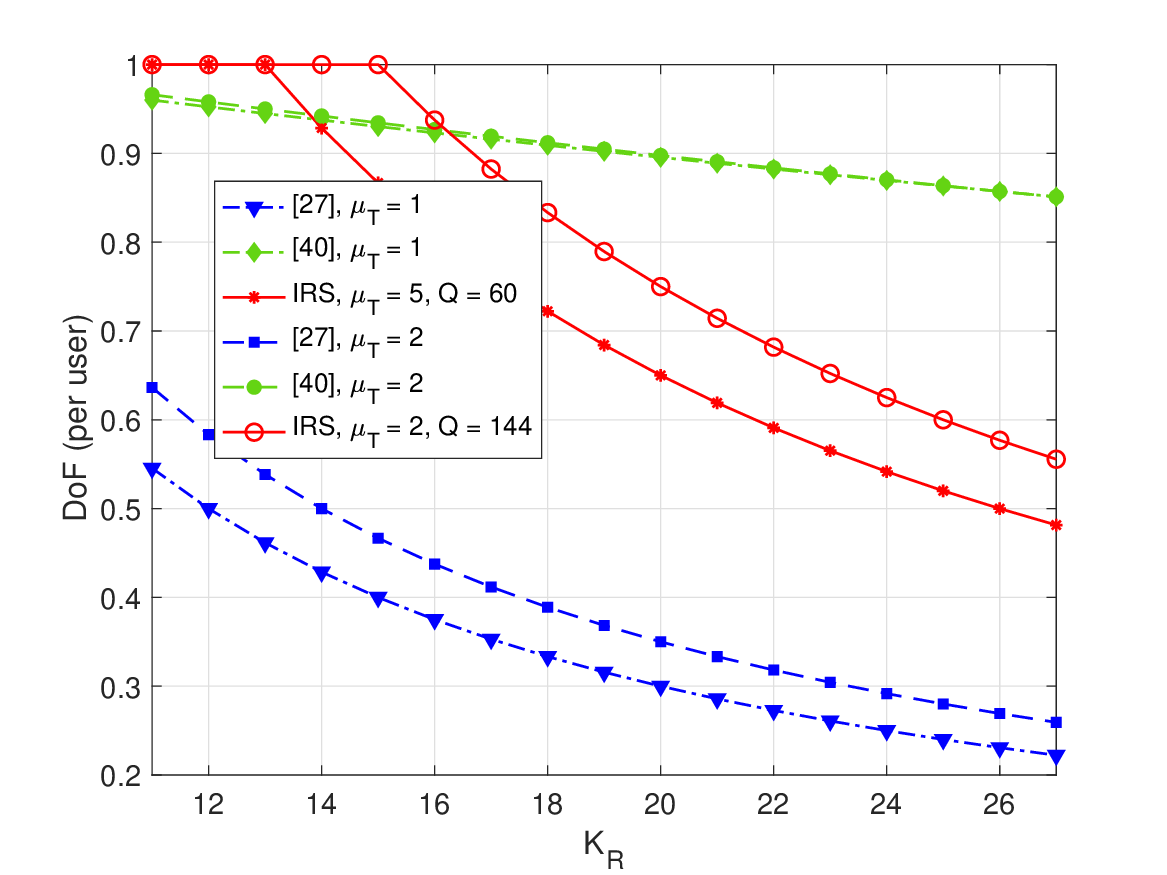}
    \caption{Achievable DoF per user versus $K_R$ for a cache-aided interference channel with $K_T = 20$ and
			${\mu_R} = 5$.}
    \label{new_fig_5}
\end{figure}

Fig. \ref{new_fig_5} depicts the DoF per user versus the number of receivers for a cache-aided interference network with $K_T = 20$ transmitters and $\mu_R = 5$. As expected, an increase in the number of receivers leads to a decrease in the DoF per user, a trend evident in all curves plotted in the figure. Notably, due to the symmetric nature of file breaking in \cite{nader} and the linearity of \eqref{maddah}, the corresponding curve experiences a more noticeable decrease compared to the asymmetric and optimal file breaking method in \cite{Tao_Fundamental}. However, an IRS can potentially support a DoF per user of $1$ up to a certain increase in the number of receivers, depending on the number of elements. The observed decrease in the curve is attributed to the insufficient number of elements relative to the number of users, which results in inadequate elimination of undesired cross-links.

\begin{figure}
    \centering
    \includegraphics[width=.7\linewidth]{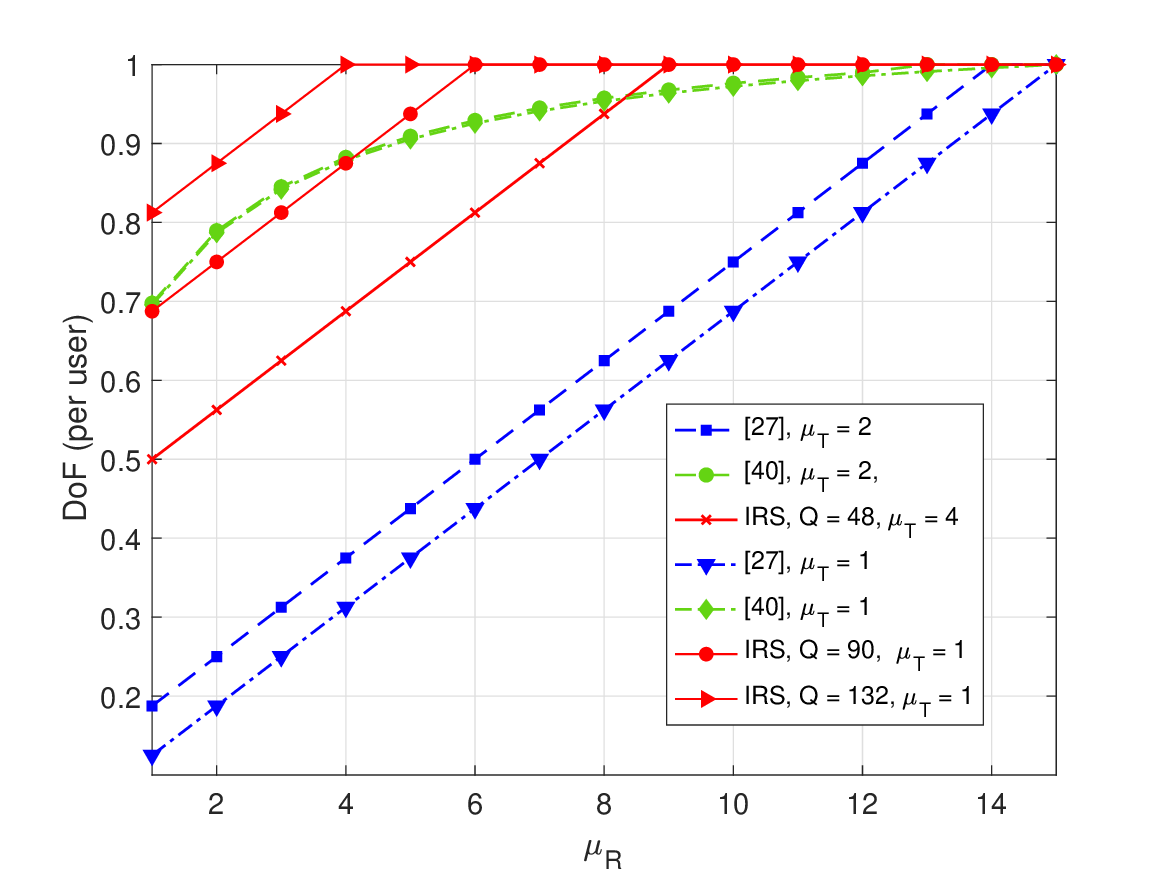}
    \caption{Achievable DoF per user versus ${\mu_R} \in                [15]$ for a cache-aided interference channel              with $K_T = K_R = 16$.}
    \label{new_fig3}
\end{figure}

Fig. \ref{new_fig3} depicts the DoF per user versus $\mu_R$ for a network with $K_T = 16$ transmitters and $K_R = 16$ receivers. 
As expected, the DoF increases with an increase in $\mu_R$. Additionally, an increase in the number of IRS elements results in a higher DoF per user, keeping the other parameters unchanged. Moreover, it is noteworthy that an IRS equipped with a sufficient number of elements provides a higher DoF compared to existing works in the literature.

		\section{Conclusion}\label{sec5con}
		In this paper, we studied the sum-DoF of a general cache-aided interference channel assisted by an active IRS. In particular, we presented a one-shot achievability scheme leveraging the transmitters’ cooperation, receivers' cache contents, interference alignment, and IRS capabilities to design the prefetching and delivery phases together with the IRS coefficients. Our results highlight the synergistic gains of jointly deploying active IRSs and coded caching, particularly in regimes where the transmitters’ cache sizes are limited. Furthermore, we have established that an active IRS with a sufficient number of elements can enable the achievement of the maximum possible sum-DoF. 

\color{black}
\section{Acknowledgment}
The authors would like to thank the anonymous reviewers for their valuable feedback, in particular for suggesting Baranyai’s theorem on hypergraph decomposition, which led to a shorter proof of Theorem 2.
\color{black}

		\section*{Appendix A: Proof of Theorem \ref{Theorem_RIS_without_extended_symbol}}\label{App:Theorem_1}
			As mentioned earlier, we consider two cases: (I) $\mu_R+L+1 \geq K_R$ and (II) $\mu_R+L+1<K_R$, and prove the theorem for each case.
			\begin{enumerate}[wide, labelwidth=!, labelindent=0pt, label=(\Roman*)]
				\item $\mu_R+L+1 \geq K_R$:\label{first_case_tau1_oneshot}
				In this case, we prove that the sum-DoF of $K_R$ is achievable. To this end, we set $L = K_R-\mu_R-1$. Based on the prefetching phase introduced in \eqref{send_general_tau1}, each requested file by receiver $j$ can be represented as 
				\begin{equation}\label{Set_requested_files_tau_1}
					{{W}}_{d_j}=\left\{W_{d_j, i, \mathcal{R}}: i \in [K_T], \mathcal{R} \subseteq\left[K_{R}\right]\backslash \{j\},|\mathcal{R}|={\mu_R}\right\}.
				\end{equation}
				
				As it can be seen, $K_T{{K_R-1} \choose {\mu_R}}$ subfiles of each requested file should be transmitted, with each transmitter having access to ${K_R-1} \choose {\mu_R}$ distinct subfiles of it. In the following, we assign each subfile ${W}_{d_j, i,  \mathcal{R}}$ in \eqref{Set_requested_files_tau_1} a symbol, denoted as  	$\tilde{\mathbf{x}}_{i,  \mathcal{R}}^{[j]} \in \mathbb{C}^{\tilde{B}}$. Next, we introduce a binary beamforming coefficient $\tilde{{{\mathbf{V}}}}_{\mathcal{R}}^{[ji]}(m)$ in each communication block $m \in [H]$ to enable a linear combination of the symbols at each transmitter. Consequently, the data stream at each transmitter $i$ at block $m \in [H]$ can be expressed as:  
				\begin{equation}
					{\mathbf{x}}^{[i]}(m) =  \sum_{j = 1}^{K_R}\sum\limits_{\substack{\mathcal{R} \subseteq [K_R]:\\|\mathcal{R}| = {\mu_R}, j \notin \mathcal{R}}}^{} {\tilde{{{\mathbf{V}}}}_{\mathcal{R}}^{[ji]}}(m){\tilde{\mathbf{x}}_{i, \mathcal{R}}^{[j]}}.
				\end{equation}
				
				Then, we set $H = K_T{{K_R-1} \choose {{\mu_R}}}$, intending to configure IRS so that each transmitted subfile is desired at one receiver and can be eliminated at the remaining receivers. Specifically, we utilize cached packets as side information and IRS‌ to eliminate the incoming interference at $\mu_R$ and $L$ receivers, respectively. As a result, after $H$ communication blocks, all the requested subfiles are received interference-free, and a DoF of $1$ is achievable at each block. We show how this can be done utilizing lemmas \ref{main_lemma_theorem_mut1_oneshot} and \ref{main_lemma_theorem_mut1_oneshot_2}. 
				
				Before presenting the lemmas, for ease of reference, we term every $K_R-1 \choose {\mu_R}$ blocks a super-block. Consequently, we have $K_T$ super-blocks, each containing $K_R-1 \choose {\mu_R}$ blocks. Moreover, for the integers $i, j,$ and $m$, we denote $i \oplus_{m} j$ as:
				\begin{equation}
					i \oplus_m j= 1 + (i+j-1 \quad \bmod m).
				\end{equation}  
				\begin{lemma}\label{main_lemma_theorem_mut1_oneshot}
					Given the prefetching phase for this case, for any receivers’ demand vector $\mathbf{d}$, the set of subfiles that need to be delivered can be partitioned into disjoint subsets of size $K_R$ as
					\begin{equation}\label{scheduling_tau1_one_shot}
						\bigcup_{\substack{
								k_2 \in \left[1:K_T\right]\\ \mathcal{R} \subseteq [K_R]\backslash\{1\}:|\mathcal{R}| = {\mu_R} }}\left\{
						W_{d_1, S_{k_1, k_2, l_1}, \mathcal{R}},
						\bigcup_{\substack{j \in \mathcal{R}}} W_{d_j, S_{k_1, k_2, l_1}, \{1\}\cup\mathcal{R}\backslash\{j\}},
						\bigcup_{\substack{j \in [K_R]\backslash\{1\}\cup\mathcal{R}}} W_{d_j, S_{k_1, k_2, l_j}, \mathcal{R}}\right\},
					\end{equation}
					where $\left\{S_{k_1, 1, l_1}, S_{k_1, 1, l_2}, \ldots, S_{k_1, 1, l_{L+1}}\right\}$ denotes arbitrary distinct indices from the set $[K_T]$ for each  $k_1 \in \left[1:{K_R-1 \choose {\mu_R}}\right]$, and we have:
					$S_{k_1, k_2, l_j} = S_{k_1, 1, l_j} \oplus_{K_T} (k_2-1)$ for $k_1 \in \left[1:{K_R-1 \choose {\mu_R}}\right], k_2 \in [2:K_T],$ and $j \in [1:L+1]$. 
				\end{lemma}
					\begin{proof}
						Given that the set \eqref{scheduling_tau1_one_shot} has a size of 
						$K_R K_T{{K_R-1} \choose {{\mu_R}}}$, representing the number of subfiles requiring delivery, it is sufficient to show that each receiver's $K_T{{K_R-1} \choose {{\mu_R}}}$ requested subfiles are included. Since $\mathcal{R} \subseteq [K_R]\backslash\{1\} \text{ s.t }|\mathcal{R}| = {\mu_R}$, the set $	\bigcup_{\substack{
								\mathcal{R} \subseteq [K_R]\backslash\{1\}:|\mathcal{R}| = {\mu_R} }}\left\{
						W_{d_1, S_{k_1, 1, l_1}, \mathcal{R}}
						\right\}$ 
						consists ${K_R - 1} \choose \mu_R$ distinct subfiles requested by receiver $1$, and as
						$S_{k_1, k_2, l_1} = S_{k_1, 1, l_1} \oplus_{K_T} (k_2-1)$ for $k_2 \in [2:K_T]$, all the $K_{T} {{K_R - 1} \choose \mu_R}$ subfiles are included. 
						On the other hand, for each receiver $j \neq 1$, the set 
						$$
						\bigcup_{\substack{
								\mathcal{R} \subseteq [K_R]\backslash\{1\}:|\mathcal{R}| = {\mu_R} }}\left\{
						\bigcup_{\substack{j \in \mathcal{R}}} W_{d_j, S_{k_1, 1, l_1}, \{1\}\cup\mathcal{R}\backslash\{j\}}
						\right\}
						$$ 
						consists ${K_R - 2} \choose {\mu_R-1}$
						and the set 
						$$
						\bigcup_{\substack{
								\mathcal{R} \subseteq [K_R]\backslash\{1\}:|\mathcal{R}| = {\mu_R} }}\left\{
						\bigcup_{\substack{j \in [K_R]\backslash\{1\}\cup\mathcal{R}}} W_{d_j, S_{k_1, 1, l_j}, \mathcal{R}}
						\right\}
						$$
						consists ${K_R - 2} \choose {\mu_R}$
						distinct subfiles of the file requested by receiver $j$, and as $S_{k_1, k_2, l_j} = S_{k_1, 1, l_j} \oplus_{K_T} (k_2-1)$ for $k_2 \in [2:K_T]$ and $j \in [L+1]$, all the $K_{T} {{K_R - 1} \choose \mu_R}$ subfiles are included. This is due to using Pascal's rule as:
						$$
						{{K_R-1} \choose {\mu_R}} = {{K_R-2} \choose {{\mu_R}-1}} + {{K_R-2} \choose {{\mu_R}}},
						$$
						and, as a result, all of the requested subfiles are partitioned in \eqref{scheduling_tau1_one_shot}.
					\end{proof}
					
					\begin{lemma}\label{main_lemma_theorem_mut1_oneshot_2}
						At any communication block $k_1 \in \left[1:{K_R-1 \choose {\mu_R}}\right]$ from any super-block $k_2 \in [K_T]$, 
						there exists a choice of beamforming coefficients, along with a network realization by IRS, such that the set of $K_R$ subfiles in 
						\begin{equation}
							\left\{
							W_{d_1, S_{k_1, k_2, l_1}, \mathcal{R}},
							\bigcup_{\substack{j \in \mathcal{R}}} W_{d_j, S_{k_1, k_2, l_1}, \{1\}\cup\mathcal{R}\backslash\{j\}},
							\bigcup_{\substack{j \in [K_R]\backslash\{1\}\cup\mathcal{R}}} W_{d_j, S_{k_1, k_2, l_j}, \mathcal{R}}\right\},
						\end{equation}
						can be delivered interference-free by the transmitters $\left\{S_{k_1, k_2, l_1}, \ldots, S_{k_1, k_2, l_{L+1}}\right\}$ where at each communication block, $\mathcal{R}$ is a distinct subset of $[K_R]\backslash \{1\}$ of size $\mu_R$. 	
					\end{lemma}
					\begin{remark}
						Note that the transmitters that are not involved in each communication block remain inactive and do not send any packets.
					\end{remark}
					\begin{proof}
						To this end, it is necessary to set the beamforming coefficient 	${\tilde{{{\mathbf{V}}}}_{\mathcal{R}}^{[ji]}}(m)$ of each selected subfiles in the transmitting set of each communication block as an identity matrix, while for the remaining subfiles, it should be set as a zero matrix. Then, we need an active IRS ‌with $Q = L(L+1)$ elements such that at each communication block 	$k_1 \in \left[1:{K_R-1 \choose {\mu_R}}\right]$ from the super-block 	$k_2 \in \left[1:K_T\right]$, the network matrix is realized in the following way:
						\begin{itemize}
							\item From the transmitter $S_{k_1, k_2, l_1}$ perspective:
							\begin{equation}
								n_{S_{k_1, k_2, l_1}, j} = 
								\left\{\begin{array}{l}
									1, \quad \forall j \in {\mathcal{R}\cup\{1\}}\\
									0, \quad \text{Otherwise}
								\end{array}\right.,
							\end{equation}
							\item From the transmitter $S_{k_1, k_2, l_j}$ perspective for  $j \in [K_R]\backslash\{1\}\cup\mathcal{R}$:
							\begin{equation}
								n_{S_{k_1, k_2, l_j}, k} = 
								\left\{\begin{array}{l}
									1, \quad \forall k \in {\mathcal{R}\cup\{j\}}\\
									0, \quad \text{Otherwise}
								\end{array}\right.,
							\end{equation}
							\item And from the 	$i \notin 	\left\{\bigcup_{\substack{j \in [1:L+1]}} S_{k_1, k_2, l_j}\right\}$ perspective:
							\begin{equation}
								n_{i, k} = 1, \quad \forall
								k \in [K_R].
							\end{equation}
						\end{itemize}
						
						To have such a realization using an active IRS, we first denote 	$\tilde{\mathcal{N}}$ as 
						\begin{equation}
							\tilde{\mathcal{N}} = \left\{(i, j) \mid i \in [K_T], j \in [K_R], n_{i, j} = 0\right\}.
						\end{equation}
						
						Now, we need to configure each element of the IRS‌ in a way that at each communication block $t$, the following equations hold:
						\begin{equation}
							\sum_{u \in\{1, \ldots, Q\}} X^{[i]}(t) H_{\mathrm{TI}}^{[u i]}(t) H_{\mathrm{IR}}^{[j u]}(t) q^{[u]}(t)=-X^{[i]}(t) H^{[j i]}(t),\quad(i, j) \in \tilde{\mathcal{N}},
						\end{equation}
						where by omitting 	$X^{[i]}(t)$, they can be written as 
						\begin{equation}\label{RIS_crosslink_removal_1_2_3}
							\sum_{u \in\{1, \ldots, Q\}}  H_{\mathrm{TI}}^{[u i]}(t) H_{\mathrm{IR}}^{[j u]}(t) q^{[u]}(t)= -H^{[j i]}(t),\quad(i, j) \in \tilde{\mathcal{N}}.
						\end{equation}
						
						This equation denotes the removal of the cross-link between transmitter $i$ and receiver $j$ and has a solution almost surely \cite[Lemma 1]{bafghi_TCom}. Hence, $L = K_R - \mu_R - 1$ unintended subfiles are not received at each receiver due to the undesired cross-links removal. Additionally, each receiver can cancel out $\mu_R$  subfiles due to its cached contents. These enable interference-free delivery of $K_R$ subfiles, each intended for a distinct receiver, at each communication block. Consequently, a DoF of $1$ is achievable at each block.
					\end{proof}

					\item $\mu_R+L+1 < K_R$: 
					In this case, we break each subfile $W_{d_j, i, \mathcal{R}}$ into ${{K_R-{\mu_R}-1} \choose L}$ smaller subfiles $W_{d_j, i, \mathcal{R}, \mathcal{L}}$, each being corresponded with a unique subset of receivers $\mathcal{R}\cup\{j\}\cup\mathcal{L}$, where $\mathcal{L} = \{\mathcal{R}_1, \ldots, \mathcal{R}_L\}$ is an arbitrary subset of size $L$ from 
					receivers ${K_R}\backslash\{j\}\cup\mathcal{R}$. 
					We assume the number of communication blocks to be $H =  K_T {{K_R} \choose {{\mu_R}+L+1}} {{{\mu_R}+L} \choose {\mu_R}}$, and, for the ease of reference, we term every ${{\mu_R}+L} \choose {\mu_R}$ communication blocks a super-block, and every $K_T$ super-blocks a hyper-block. Therefore, we have $K_R \choose {{\mu_R}+L+1}$ hyper-blocks.
					In the following, we consider a subset $\mathcal{C} \subseteq [K_R]$ of size ${\mu_R}+L+1$. Then, for the sub-network of $[K_T] \times \mathcal{C}$, we configure the IRS‌ and schedule the delivery of the subfiles as explained in the case \ref{first_case_tau1_oneshot} It is worth mentioning that the sub-network of $[K_T] \times [K_R]\backslash\mathcal{C}$ is supposed to be fully-connected, and, as a result, there is no need for further IRS configuration. Therefore, the number of IRS‌ elements needs to be $Q = L(L+1)$. 
					
					Now, it is easy to see that the DoF of $1$ is achievable for each of the receivers in the subset of $\mathcal{C}$ for ${K_T} {{{\mu_R}+L} \choose {\mu_R}}$ communication blocks of a hyper-block, and the DoF of all the receivers $[K_R]\backslash\mathcal{C}$ is equal to zero. Subsequently, it is sufficient to repeat this procedure for other subsets of size ${\mu_R}+L+1$ from the receivers. To calculate the DoF of the system, note that each receiver can decode its desired subfile interference-free in ${K_T}{{K_R-1} \choose {{\mu_R}+L}}{{{\mu_R}+L} \choose {{\mu_R}}}$ communication blocks and the number of all communication blocks is $K_T {{K_R} \choose {L+{\mu_R}+1}}{{{\mu_R}+L} \choose {{\mu_R}}}$. Consequently, the DoF of
					$$\mathrm{DoF_{user}} = \frac{{K_T}{{K_R-1} \choose {L+{\mu_R}}}{{{\mu_R}+L} \choose {{\mu_R}}}}{K_T {{K_R} \choose {L+{\mu_R}+1}}{{{\mu_R}+L} \choose {{\mu_R}}}} = \frac{L+{\mu_R}+1}{K_R},$$
					is achievable for each receiver, and the sum-DoF of $L+{\mu_R}+1$ can be achieved in this case, which completes the proof of Theorem \ref{Theorem_RIS_without_extended_symbol}.
				\end{enumerate}
			
			\section*{Appendix B: Proof of Theorem \ref{Theorem_RIS_without_extended_symbol_tau_geq_1}}\label{App:Theorem_2}
			As mentioned earlier, we consider two cases: (I) $\mu_T+\mu_R+L \geq K_R$ and (II) $\mu_T+\mu_R+L <K_R$, and prove the theorem for each case, separately.
			\begin{enumerate}[wide, labelwidth=!, labelindent=0pt, label=(\Roman*)]
				\item $\mu_R+\mu_T+L \geq K_R$:
				\label{first_case_taugeq1_one_shot}
				In this case, we prove that the sum-DoF of $K_R$ is achievable. To this end, we set $L = K_R - (\mu_R+\mu_T)$. Based on the prefetching phase introduced in \eqref{send_general_tau1}, 	${K_T \choose {\mu_T}}{{K_R-1} \choose {\mu_R}}$ subfiles of each requested file should be transmitted, where each ${K_R-1} \choose {\mu_R}$ subfiles of these subfiles are available at $\mu_T$ transmitters. Now, we break each subfile 		$W_{d_j, \mathcal{S}, \mathcal{R}}$  into 		${K_R-{\mu_R}-1} \choose {{\mu_T}-1}$ smaller subfiles 		$W_{d_j, \mathcal{S}, \mathcal{R}, \mathcal{T}}$, each being corresponded with a unique subset of receivers 		$\mathcal{R}\cup\{j\}\cup\mathcal{T}$, where 		$\mathcal{T} = \{\mathcal{R}_1, \ldots, \mathcal{R}_{{\mu_T}-1}\}$ is an arbitrary subset of size $\mu_T-1$ from the receivers $[K_R]\backslash\{j\}\cup\mathcal{R}$. This allows each subfile $W_{d_j, \mathcal{S}, \mathcal{R}, \mathcal{T}}$ to be eliminated not only in the receivers $\mathcal{R}$ due to side information but also in the receivers $\mathcal{T}$ due to the possibility of cooperation between the transmitters, as there exist overlaps between the transmitters' cache memories. 
				
				Now, we assign each subfile ${W}_{d_j, \mathcal{S},  \mathcal{R}, \mathcal{T}}$ a symbol, denoted by 		$\tilde{\mathbf{x}}_{\mathcal{S},  \mathcal{R}, \mathcal{T}}^{[j]} \in \mathbb{C}^{\tilde{B}}$.  
				Then, for each of these symbols, we introduce a complex beamforming coefficient 	$\tilde{{{\mathbf{V}}}}_{\mathcal{S}, \mathcal{R}, \mathcal{T}}^{[ji]}(m) \in \mathbb{C}^{\tilde{B} \times \tilde{B}}$ at communication block $m \in [H]$ to enable the possibility of transmitting a linear combination of the symbols at each transmitter $i \in [M\mu_T]$. Consequently, the data stream at each transmitter $i$ at block $m \in [H]$ can be expressed as:
				\begin{equation}\label{TX_i_tau_geq_1}
					{\mathbf{x}}^{[i]}(m) =  \sum_{j = 1}^{K_R}\sum\limits_{\substack{\mathcal{R} \subseteq [K_R]:\\|\mathcal{R}| = {\mu_R}, j \notin \mathcal{R}}}^{} 
					\sum\limits_{\substack{\mathcal{S} \subseteq [K_T]:\\|\mathcal{S}| = {\mu_T}, i \in \mathcal{S}}}^{}
					\sum\limits_{\substack{\mathcal{T} \subseteq [K_R]\backslash\{j\}\cup\mathcal{R}:\\|\mathcal{T}| = {\mu_T}-1}}^{}
					{\tilde{{{\mathbf{V}}}}_{\mathcal{S}, \mathcal{R}, \mathcal{T}}^{[ji]}}(m){\tilde{\mathbf{x}}_{\mathcal{S}, \mathcal{R}, \mathcal{T}}^{[j]}}.
				\end{equation}		
				
				In the following, we first present Lemmas \ref{lemma_1RX_without_interference} and \ref{lemma_tplustau_RX_without_interference}, introducing two forms of subfiles which can be received interference-free at the intended receivers. 
				
				\begin{lemma}
					\label{lemma_1RX_without_interference}	
					In a cache-aided network composed of transmitters 			$\mathcal{S} = \{S_i\}_{i = 1}^{{\mu_T}}$ and receivers 		$\mathcal{R} = \{R_j\}_{j = 1}^{{\mu_R}+{\mu_T}}$, it is possible to decode the subfile 		$W_{d_j, \mathcal{S}, \mathcal{R}_j, \mathcal{T}_j}$ interference-free at one communication block in the receiver ${R}_j$ for arbitrary 			$\mathcal{R}_j, \mathcal{T}_j \subseteq \mathcal{R}\backslash\{R_j\}$, where 	$|\mathcal{R}_j| = {\mu_R}$,
					$|\mathcal{T}_j| = {\mu_T}-1$, and
					$\mathcal{R}_j \cup \mathcal{T}_j = \mathcal{R}\backslash\{R_j\}$.
				\end{lemma}	
				\begin{proof}
					The data stream at transmitter $S_i \in \mathcal{S}$ in a communication block can be written as:
					\begin{equation}
						{\mathbf{x}}^{[i]} =  
						{\tilde{{\mathbf{V}}}_{\mathcal{S}, \mathcal{R}_j, \mathcal{T}_j}^{[ji]}}{\tilde{\mathbf{x}}_{\mathcal{S}, \mathcal{R}_j, \mathcal{T}_j}^{[j]}},
					\end{equation}		
					where 	${\tilde{{\mathbf{V}}}_{\mathcal{S}, \mathcal{R}_j, \mathcal{T}_j}^{[ji]}} \in \mathbb{C}^{\tilde{B} \times \tilde{B}}$ and ${\tilde{\mathbf{x}}_{\mathcal{S}, \mathcal{R}_j, \mathcal{T}_j}^{[j]}} \in \mathbb{C}^{\tilde{B} \times 1}$ are the diagonal beamforming matrix and the corresponding symbol at transmitter $i$ for the subfile $W_{d_j, \mathcal{S}, \mathcal{R}_j, \mathcal{T}_j}$. In the following, we demonstrate how the beamforming matrix should be determined to prove the lemma.
					
					The received signal at the receiver $R_j$ is:
					{\begin{align}
							{\mathbf{y}}^{[j]} & = \sum_{i = 1}^{{\mu_T}} \tilde{{\mathbf{H}}}^{[ji]} {{\mathbf{x}}}^{[i]} = \sum_{i = 1}^{{\mu_T}} \tilde{{\mathbf{H}}}^{[ji]} {\tilde{{{\mathbf{V}}}}_{\mathcal{S}, \mathcal{R}_j, \mathcal{T}_j}^{[ji]}}{\tilde{\mathbf{x}}_{\mathcal{S}, \mathcal{R}_j, \mathcal{T}_j}^{[j]}},
					\end{align}}
					and at receiver 
					$R_k$ for $k \in [{\mu_R}+{\mu_T}]\backslash\{j\}$, we have:
					{\begin{align}
							{\mathbf{y}}^{[k]} & = \sum_{i = 1}^{{\mu_T}} \tilde{{\mathbf{H}}}^{[ki]} {{\mathbf{x}}}^{[i]} = \sum_{i = 1}^{{\mu_T}} \tilde{{\mathbf{H}}}^{[ki]} {\tilde{{\mathbf{V}}}_{\mathcal{S}, \mathcal{R}_j, \mathcal{T}_j}^{[ji]}}{\tilde{\mathbf{x}}_{\mathcal{S}, \mathcal{R}_j, \mathcal{T}_j}^{[j]}}.
					\end{align}}
					
					Note that at the receiver $R_k \in \mathcal{R}_j$, the incoming interference caused by the symbol ${\tilde{\mathbf{x}}_{\mathcal{S}, \mathcal{R}_j, \mathcal{T}_j}^{[j]}}$ can be eliminated because of its cache contents. Now, to decode this symbol at $R_j$ without causing interference for the receivers $R_k \in \mathcal{R}\backslash\mathcal{R}_j$, the following equations should be satisfied:
					\begin{align}
						&\sum_{i = 1}^{{\mu_T}} \tilde{{\mathbf{H}}}^{[ji]} {\tilde{{\mathbf{V}}}_{\mathcal{S}, \mathcal{R}_j, \mathcal{T}_j}^{[ji]}}{\tilde{\mathbf{x}}_{\mathcal{S}, \mathcal{R}_j, \mathcal{T}_j}^{[j]}} = 1, \label{1st_lemma_decode_constraint}
						\\
						&\sum_{i = 1}^{{\mu_T}} \tilde{{\mathbf{H}}}^{[ki]} {\tilde{{\mathbf{V}}}_{\mathcal{S}, \mathcal{R}_j, \mathcal{T}_j}^{[ji]}}{\tilde{\mathbf{x}}_{\mathcal{S}, \mathcal{R}_j, \mathcal{T}_j}^{[j]}} = 0,\quad \forall k \in \mathcal{T}_j. \label{1st_lemma_decode_2nd_constraint}
					\end{align}
					
					As it can be observed, the equations	\eqref{1st_lemma_decode_constraint}-\eqref{1st_lemma_decode_2nd_constraint} form a system of 				$\tilde{B}{\mu_T}$ linear equations. On the other hand, the number of  variables $\{\tilde{{v}}_{\mathcal{S}, \mathcal{R}_j, \mathcal{T}_j}^{[ji]}\}$
					are equal to
					$\tilde{B}{\mu_T}$. This implies that there is always a set of coefficients such that the equations \eqref{1st_lemma_decode_constraint}-\eqref{1st_lemma_decode_2nd_constraint} are satisfied.
				\end{proof}
				\begin{lemma}
					\label{lemma_tplustau_RX_without_interference}
					Consider a cache-aided network composed of transmitters		$\mathcal{S} = \{S_i\}_{i = 1}^{{\mu_T}}$ and receivers 	$\mathcal{R} = \{R_j\}_{j = 1}^{{\mu_R}+{\mu_T}}$. Then, denote $R_{\mathcal{I}} \triangleq  \{R_{i_1}, R_{i_2}, \ldots, R_{i_I}\}$ for 	$\mathcal{I} = \{i_1, i_2, \ldots, i_I\}$. By suitable design of the beamforming coefficients at the transmitters, it is possible to decode the $\mu_T+\mu_R$ subfiles
					{\small \begin{align}
							\biggl\{	
							\bigcup_{\substack{j \in [1:{\mu_R}+1]}}  \left\{ W_{d_j, \mathcal{S}, R_{[1:{\mu_R}+1]\backslash\{j\}}, R_{[{\mu_R}+2:{\mu_R}+{\mu_T}]}}\right\},
							\bigcup_{\substack{j \in [{\mu_R}+2:{\mu_R}+{\mu_T}]}}  \left\{ W_{d_j, \mathcal{S}, R_{[2:{\mu_R}+1]}, R_{\{1\}\cup[{\mu_R}+2:{\mu_R}+{\mu_T}]\backslash\{j\}}}\right\}
							\biggr\},
					\end{align}}
					interference-free at the corresponding receivers in one communication block.
				\end{lemma}	
				\begin{proof}
					The data stream at the transmitter $i$‌ can be written as:
					\begin{align}
						{\mathbf{x}}^{[i]} & =  \sum_{j = 1}^{{\mu_R}+1}
						{\tilde{{\mathbf{V}}}_{\mathcal{S}, R_{[1:{\mu_R}+1]\backslash\{j\}}, R_{[{\mu_R}+2:{\mu_R}+{\mu_T}]}}^{[ji]}}{\tilde{\mathbf{x}}_{\mathcal{S}, R_{[1:{\mu_R}+1]\backslash\{j\}}, R_{[{\mu_R}+2:{\mu_R}+{\mu_T}]}}^{[j]}} \nonumber \\
						& +
						\sum_{j = {\mu_R}+2}^{{\mu_R}+{\mu_T}}
						{\tilde{{\mathbf{V}}}_{\mathcal{S}, R_{[2:{\mu_R}+1]}, R_{\{1\}\cup[{\mu_R}+2:{\mu_R}+{\mu_T}]\backslash\{j\}}}^{[ji]}}{\tilde{\mathbf{x}}_{\mathcal{S}, R_{[2:{\mu_R}+1]}, R_{\{1\}\cup[{\mu_R}+2:{\mu_R}+{\mu_T}]\backslash\{j\}}}^{[j]}}, 
					\end{align}		
					where, in the following, we demonstrate how the beamforming matrix should be determined to prove the lemma. 
					
					The received signal at the receiver 				$R_j \in \mathcal{R}$ can be written as:
					\begin{align}\label{general_term_rx_tau_geq1}
						{\mathbf{y}}^{[j]} & = \sum_{i = 1}^{{\mu_T}} \tilde{{\mathbf{H}}}^{[ji]} {{\mathbf{x}}}^{[i]} = \sum_{i = 1}^{{\mu_T}} \tilde{{\mathbf{H}}}^{[ji]}  \biggl(\sum_{k = 1}^{{\mu_R}+1}
						{\tilde{{\mathbf{V}}}_{\mathcal{S}, R_{[1:{\mu_R}+1]\backslash\{k\}}, R_{[{\mu_R}+2:{\mu_R}+{\mu_T}]}}^{[ki]}}{\tilde{\mathbf{x}}_{\mathcal{S}, R_{[1:{\mu_R}+1]\backslash\{k\}}, R_{[{\mu_R}+2:{\mu_R}+{\mu_T}]}}^{[k]}} \nonumber \\
						& +
						\sum_{k = {\mu_R}+2}^{{\mu_R}+{\mu_T}}
						{\tilde{{\mathbf{V}}}_{\mathcal{S}, R_{[2:{\mu_R}+1]}, R_{\{1\}\cup[{\mu_R}+2:{\mu_R}+{\mu_T}]\backslash\{k\}}}^{[ki]}}{\tilde{\mathbf{x}}_{\mathcal{S}, R_{[2:{\mu_R}+1]}, R_{\{1\}\cup[{\mu_R}+2:{\mu_R}+{\mu_T}]\backslash\{k\}}}^{[k]}}\biggr),
					\end{align}
					where the involving terms can be divided into two categories including desired and undesired, based on the index $j$:
					\begin{itemize}[wide, labelwidth=!, labelindent=0pt]
						\item $j = 1$: 
						In this case, the equation 						\eqref{general_term_rx_tau_geq1} for the receiver 
						$R_1$ can be rewritten as:
						\begin{align}
							{\mathbf{y}}^{[1]} & = \underbrace{\sum_{i = 1}^{{\mu_T}} \tilde{{\mathbf{H}}}^{[1i]} 
								{\tilde{{\mathbf{V}}}_{\mathcal{S}, R_{[1:{\mu_R}+1]\backslash\{1\}}, R_{[{\mu_R}+2:{\mu_R}+{\mu_T}]}}^{[1i]}}{\tilde{\mathbf{x}}_{\mathcal{S}, R_{[1:{\mu_R}+1]\backslash\{1\}}, R_{[{\mu_R}+2:{\mu_R}+{\mu_T}]}}^{[1]}}}_{\text{desired term}} \nonumber \\
							& + \underbrace{\sum_{i = 1}^{{\mu_T}} \tilde{{\mathbf{H}}}^{[1i]} \sum_{k = 2}^{{\mu_R}+1} 
								{\tilde{{\mathbf{V}}}_{\mathcal{S}, R_{[1:{\mu_R}+1]\backslash\{k\}}, R_{[{\mu_R}+2:{\mu_R}+{\mu_T}]}}^{[ki]}}{\tilde{\mathbf{x}}_{\mathcal{S}, R_{[1:{\mu_R}+1]\backslash\{k\}}, R_{[{\mu_R}+2:{\mu_R}+{\mu_T}]}}^{[k]}}}_{\text{undesired terms $1$}} \nonumber \\
							& + \underbrace{\sum_{i = 1}^{{\mu_T}} \tilde{{\mathbf{H}}}^{[1i]} 
								\sum_{k = {\mu_R}+2}^{{\mu_R}+{\mu_T}}
								{\tilde{{\mathbf{V}}}_{\mathcal{S}, R_{[2:{\mu_R}+1]}, R_{\{1\}\cup[{\mu_R}+2:{\mu_R}+{\mu_T}]\backslash\{k\}}}^{[ki]}}{\tilde{\mathbf{x}}_{\mathcal{S}, R_{[2:{\mu_R}+1]}, R_{\{1\}\cup[{\mu_R}+2:{\mu_R}+{\mu_T}]\backslash\{k\}}}^{[k]}}}_{\text{undesired terms $2$}}. \label{linear_eq_Rx1}
						\end{align}
						
						As it can be seen, the equation	\eqref{linear_eq_Rx1} is composed of three terms: a desired term that indicates the desired subfile by this receiver and undesired terms $1$ and $2$ that include the subfiles intended at other receivers and should be eliminated at this receiver. Note that the undesired terms $1$ can easily be eliminated using the cache contents at this receiver. To decode the intended subfile of this receiver, as well as to eliminate the undesired terms $2$, the following equations should be satisfied:
						{\begin{align}
								& \sum_{i = 1}^{{\mu_T}} \tilde{{\mathbf{H}}}^{[1i]} 
								{\tilde{{\mathbf{V}}}_{\mathcal{S}, R_{[1:{\mu_R}+1]\backslash\{1\}}, R_{[{\mu_R}+2:{\mu_R}+{\mu_T}]}}^{[1i]}}{\tilde{\mathbf{x}}_{\mathcal{S}, R_{[1:{\mu_R}+1]\backslash\{1\}}, R_{[{\mu_R}+2:{\mu_R}+{\mu_T}]}}^{[1]}} = 1, \label{desired_R1}\\
								& \sum_{i = 1}^{{\mu_T}} \tilde{{\mathbf{H}}}^{[1i]} 
								{\tilde{{\mathbf{V}}}_{\mathcal{S}, R_{[2:{\mu_R}+1]}, R_{\{1\}\cup[{\mu_R}+2:{\mu_R}+{\mu_T}]\backslash\{k\}}}^{[ki]}}{\tilde{\mathbf{x}}_{\mathcal{S}, R_{[2:{\mu_R}+1]}, R_{\{1\}\cup[{\mu_R}+2:{\mu_R}+{\mu_T}]\backslash\{k\}}}^{[k]}} = 0, \nonumber \\
								& \forall k \in [{\mu_R}+2:{\mu_R}+{\mu_T}], \label{undesired_R1}
						\end{align}}
						where, in total, includes 	${\tilde{B}}{\mu_T}$ linear equations.
						\item $j \in [2: \mu_R+1]$: 
						Similarly, the equation 	\eqref{general_term_rx_tau_geq1} for this case be rewritten as:
						\begin{align}
							{\mathbf{y}}^{[j]} & = \underbrace{\sum_{i = 1}^{{\mu_T}} \tilde{{\mathbf{H}}}^{[ji]} 
								{\tilde{\mathbf{V}}_{\mathcal{S}, R_{[1:{\mu_R}+1]\backslash\{j\}}, R_{[{\mu_R}+2:{\mu_R}+{\mu_T}]}}^{[ji]}}{\tilde{\mathbf{x}}_{\mathcal{S}, R_{[1:{\mu_R}+1]\backslash\{1\}}, R_{[{\mu_R}+2:{\mu_R}+{\mu_T}]}}^{[j]}}}_{\text{desired term}} \nonumber \\
							& + \underbrace{\sum_{i = 1}^{{\mu_T}} \tilde{\mathbf{H}}^{[ji]} \sum_{\substack{k = 1, \\ k \neq j}}^{{\mu_R}+1} 
								{\tilde{\mathbf{V}}_{\mathcal{S}, R_{[1:{\mu_R}+1]\backslash\{k\}}, R_{[{\mu_R}+2:{\mu_R}+{\mu_T}]}}^{[ki]}}{\tilde{\mathbf{x}}_{\mathcal{S}, R_{[1:{\mu_R}+1]\backslash\{k\}}, R_{[{\mu_R}+2:{\mu_R}+{\mu_T}]}}^{[k]}}}_{\text{undesired terms $1$}} \nonumber \\
							& + \underbrace{\sum_{i = 1}^{{\mu_T}} \tilde{\mathbf{H}}^{[ji]} 
								\sum_{k = {\mu_R}+2}^{{\mu_R}+{\mu_T}}
								{\tilde{\mathbf{V}}_{\mathcal{S}, R_{[2:{\mu_R}+1]}, R_{\{1\}\cup[{\mu_R}+2:{\mu_R}+{\mu_T}]\backslash\{k\}}}^{[ki]}}{\tilde{\mathbf{x}}_{\mathcal{S}, R_{[2:{\mu_R}+1]}, R_{\{1\}\cup[{\mu_R}+2:{\mu_R}+{\mu_T}]\backslash\{k\}}}^{[k]}}}_{\text{undesired terms $2$}}.\label{linear_eq_Rx_2_tpluas1}
						\end{align}
						
						As it can be seen, the equation	\eqref{linear_eq_Rx_2_tpluas1} is composed of three terms: a desired term that indicates the desired subfile by this receiver and undesired terms $1$ and $2$ that includes the subfiles intended at other receivers and should be eliminated at this receiver. Note that all of the subfiles corresponding to the undesired terms $1$ and $2$ have been already cached in the prefetching phase in this receiver and, therefore, can be easily eliminated. Hence, to decode the intended subfile of this receiver interference-free, the following equations should be satisfied:
						{\begin{align}
								\sum_{i = 1}^{{\mu_T}} \tilde{\mathbf{H}}^{[ji]} 
								{\tilde{\mathbf{V}}_{\mathcal{S}, R_{[1:{\mu_R}+1]\backslash\{j\}}, R_{[{\mu_R}+2:{\mu_R}+{\mu_T}]}}^{[ji]}}{\tilde{\mathbf{x}}_{\mathcal{S}, R_{[1:{\mu_R}+1]\backslash\{j\}}, R_{[{\mu_R}+2:{\mu_R}+{\mu_T}]}}^{[1]}} = 1,\quad j \in [2:{\mu_R}+1],\label{desired_R2_Rtplus1}
						\end{align}}
						where, in total, includes 	${\tilde{B}}{\mu_R}$ linear equations.
						\item $j \in [{\mu_R}+2:{\mu_R}+{\mu_T}]$: The equation 			\eqref{general_term_rx_tau_geq1} at the receiver $R_j$‌ can be rewritten as:‌
						\begin{align}
							{\mathbf{y}}^{[j]} & = \underbrace{\sum_{i = 1}^{{\mu_T}} \tilde{\mathbf{H}}^{[ji]} 
								{\tilde{\mathbf{V}}_{\mathcal{S}, R_{[2:{\mu_R}+1]}, R_{\{1\}\cup[{\mu_R}+2:{\mu_R}+{\mu_T}]\backslash\{j\}}}^{[ji]}}{\tilde{\mathbf{x}}_{\mathcal{S}, R_{[2:{\mu_R}1]}, R_{\{1\}\cup[{\mu_R}+2:{\mu_R}+{\mu_T}]\backslash\{j\}}}^{[j]}}}_{\text{desired term}} \nonumber \\
							& + \underbrace{\sum_{i = 1}^{{\mu_T}} \tilde{\mathbf{H}}^{[ji]} \sum_{\substack{k = 1}}^{{\mu_R}+1} 
								{\tilde{\mathbf{V}}_{\mathcal{S}, R_{[1:{\mu_R}+1]\backslash\{k\}}, R_{[{\mu_R}+2:{\mu_R}+{\mu_T}]}}^{[ki]}}{\tilde{\mathbf{x}}_{\mathcal{S}, R_{[1:{\mu_R}+1]\backslash\{k\}}, R_{[{\mu_R}+2:{\mu_R}+{\mu_T}]}}^{[k]}}}_{\text{undesired terms $1$}} \nonumber \\
							& + \underbrace{\sum_{i = 1}^{{\mu_T}} \tilde{\mathbf{H}}^{[ji]} 
								\sum_{\substack{k = {\mu_R}+2, \\ k \neq j}}^{{\mu_R}+{\mu_T}}
								{\tilde{\mathbf{V}}_{\mathcal{S}, R_{[2:{\mu_R}+1]}, R_{\{1\}\cup[{\mu_R}+2:{\mu_R}+{\mu_T}]\backslash\{k\}}}^{[ki]}}{\tilde{\mathbf{x}}_{\mathcal{S}, R_{[2:{\mu_R}+1]}, R_{\{1\}\cup[{\mu_R}+2:{\mu_R}+{\mu_T}]\backslash\{k\}}}^{[k]}}}_{\text{undesired terms $2$}}.\label{linear_eq_Rx_tplus2_tpluastau}
						\end{align}
						As it can be observed, the equation							\eqref{linear_eq_Rx_tplus2_tpluastau} is composed of three terms: a desired term that indicates the desired subfile by this receiver and undesired terms $1$ and $2$ that include the subfiles intended at other receivers and should be eliminated at this receiver. Note that none of the subfiles corresponding to undesired terms $1$ and $2$ have been cached in the prefetching phase in this receiver. Hence, to decode the corresponding subfile of this receiver interference-free, the following equations should be satisfied:
						{\begin{align}
								& \sum_{i = 1}^{{\mu_T}} \tilde{\mathbf{H}}^{[ji]} 
								{\tilde{\mathbf{V}}_{\mathcal{S}, R_{[2:{\mu_R}+1]}, R_{\{1\}\cup[{\mu_R}+2:{\mu_R}+{\mu_T}]\backslash\{j\}}}^{[ji]}}{\tilde{\mathbf{x}}_{\mathcal{S}, R_{[2:{\mu_R}+1]}, R_{\{1\}\cup[{\mu_R}+2:{\mu_R}+{\mu_T}]\backslash\{j\}}}^{[j]}} = 1,\nonumber\\
								& \forall j \in [{\mu_R}+2:{\mu_R}+{\mu_T}], \label{desired_Rtplus2_Rtplusttau} \\
								& \sum_{i = 1}^{{\mu_T}} \tilde{\mathbf{H}}^{[ji]}
								{\tilde{\mathbf{V}}_{\mathcal{S}, R_{[1:{\mu_R}+1]\backslash\{k\}}, R_{[{\mu_R}+2:{\mu_R}+{\mu_T}]}}^{[ki]}}{\tilde{\mathbf{x}}_{\mathcal{S}, R_{[1:t+1]\backslash\{k\}}, R_{[{\mu_R}+2:{\mu_R}+{\mu_T}]}}^{[k]}} = 0, \nonumber \\
								& \forall k \in [1:{\mu_R}+1], \quad \forall j \in [{\mu_R}+2:{\mu_R}+{\mu_T}],\label{undesired_1_Rtplus2_Rtplusttau}\\
								& \sum_{i = 1}^{{\mu_T}} \tilde{\mathbf{H}}^{[ji]} 
								{\tilde{\mathbf{V}}_{\mathcal{S}, R_{[2:{\mu_R}+1]}, R_{\{1\}\cup[{\mu_R}+2:{\mu_R}+{\mu_T}]\backslash\{k\}}}^{[ki]}}{\tilde{\mathbf{x}}_{\mathcal{S}, R_{[2:{\mu_R}+1]}, R_{\{1\}\cup[{\mu_R}+2:{\mu_R}+{\mu_T}]\backslash\{k\}}}^{[k]}} = 0, \nonumber \\
								& \forall j, k \in [{\mu_R}+2:{\mu_R}+{\mu_T}], \quad k \neq j,\label{undesired_2_Rtplus2_Rtplusttau}
						\end{align}}
						where, in total, include	$\tilde{B}({\mu_T}-1)({\mu_R}+{\mu_T})$ linear equations.
						
						As it can be seen, the equations 
						\eqref{desired_R1},
						\eqref{undesired_R1},
						\eqref{desired_R2_Rtplus1},
						\eqref{desired_Rtplus2_Rtplusttau},
						\eqref{undesired_1_Rtplus2_Rtplusttau},
						and
						\eqref{undesired_2_Rtplus2_Rtplusttau} introduce a system of 	$\tilde{B}{\mu_T}({\mu_R}+{\mu_T})$ linear equations. On the other hand, the number of variables is equal to 	$\tilde{B}{\mu_T}({\mu_R}+{\mu_T})$.  This implies that there is always a set of coefficients such that the aforementioned equations \eqref{1st_lemma_decode_constraint}-\eqref{1st_lemma_decode_2nd_constraint} are satisfied.
					\end{itemize}
				\end{proof}
We now present two proofs: one that uses Baranyai's theorem and another that does not.
				\renewcommand{\labelenumii}{\theenumii}
				\renewcommand{\theenumii}{\theenumi\Alph{enumii}.}
				\begin{enumerate}
					[wide, labelwidth=!, labelindent=0pt]
					\item Proof using Baranyai's theorem:
					In this case, we assume the number of communication blocks as 	$H = {M{\mu_T} \choose {{\mu_T}}}{{K_R-1} \choose {{\mu_R}}}{{K_R-{\mu_R}-1} \choose {{\mu_T} - 1}}$. Then, we configure the network matrix such that each transmitted subfile can be successfully decoded at the intended receiver, eliminated in $\mu_R$ receivers due to the receivers' cache contents, in $\mu_T-1$ receivers due to utilizing the ZF technique, and in L receivers due to the IRS‌ configuration. Hence, after $H$ communication blocks, all the requested subfiles are decoded successfully at the intended receivers, and the DoF of $1$ can be achievable for each receiver, shown in Lemmas \ref{lemma1_tau_geq1}-\ref{lemma2_tau_geq1}. 
					
					Here, for ease of reference, we term every 	${K_R-1 \choose {\mu_R}}{K_R-{\mu_R}-1 \choose {\mu_T}-1}$ blocks a super-block and every $M$ super-blocks a hyper-block. Therefore, we have 
					${M{\mu_T}-1 \choose {\mu_T}-1}$ hyper-blocks.
					Given the existence of an $(M, \mu_T)$-subset partition by Baranyai's theorem, we number each of the composing subsets from $1$ to ${M\mu_T} \choose \mu_T$. This is done in such a way that each of the  sets of numbers	$[1:M]$,
					$[M+1:2M]$,
					$\ldots$,
					$\left[{M{\mu_T} \choose {{\mu_T}}}-M+1:{M{\mu_T} \choose {{\mu_T}}}\right]$
					is assigned to the composing subsets of each of the distinct partitions. The number assigned to the set $S$ is denoted by $\kappa(S)$.
					
					\begin{lemma}
						\label{lemma1_tau_geq1}
						Given the prefetching phase for this case, for any receivers’ demand vector $\mathbf{d}$, the set of subfiles that need to be delivered can be partitioned into disjoint subsets of size $K_R$ as
						\begin{align}\label{scheduling_tau_geq1}
							\bigcup_{\substack{
									k_2 \in \left[1:M\right]\\ 
									k_3 \in \left[1:{1:{{M\mu_T-1} \choose {\mu_T-1}}}\right]\\
									\mathcal{R} \subseteq [K_R]\backslash\{1\}:|\mathcal{R}| = {\mu_R}\\
									\mathcal{T} \subseteq [K_R]\backslash \mathcal{R}\cup\{1\}: |\mathcal{T}| = {\mu_T} -1	
							}}
							\Biggl\{&
							W_{d_1, S_{k_1, k_2, k_3, l_1}, \mathcal{R}, \mathcal{T}},
							\bigcup_{\substack{j \in \mathcal{R}}} W_{d_j, S_{k_1, k_2, k_3, l_1}, \{1\}\cup\mathcal{R}\backslash\{j\}, \mathcal{T}}, \nonumber \\
							&
							\bigcup_{\substack{j \in \mathcal{T}}} W_{d_j, S_{k_1, k_2, k_3, l_1}, \mathcal{R}, \{1\}\cup\mathcal{T}\backslash\{j\}},  \bigcup_{\substack{j \in [K_R]\backslash\{1\}\cup\mathcal{R}\cup\mathcal{T}}} W_{d_j, S_{k_1, k_2, k_3, l_j}, \mathcal{R}, \mathcal{T}}
							\Biggr\},
						\end{align}
						where $\left\{S_{k_1, 1, 1, l_1}, S_{k_1, 1, 1, l_2}, \ldots, S_{k_1, 1, 1, l_{L+1}}\right\}$ denotes arbitrary distinct indices from an arbitrary set in the form 
						\begin{equation}
							[kM-M+1:kM], \quad 
							k \in \left[1:{M{\mu_T}-1 \choose {{\mu_T}-1}}\right],
						\end{equation}
						for each  $k_1 \in \left[1:{K_R-1 \choose {\mu_R}}{K_R-\mu_R-1 \choose {\mu_T-1}}\right]$, and we have:
						\begin{equation}
							\kappa(S_{k_1, k_2, k_3, l_j}) = \kappa(S_{k_1, 1, 1, l_j}) \oplus_{M} (k_2-1) + (k_3-1)M, \quad (k_2, k_3) \neq (1, 1),
						\end{equation}
						for
						$k_1 \in \left[1:{K_R-1 \choose {\mu_R}}{K_R-{\mu_R}-1 \choose {{\mu_T}-1}}\right]$, 
						$k_2 \in [1:M]$,
						$k_3 \in [1:{M{\mu_T}-1 \choose {{\mu_T}-1}}]$, and 
						$j \in [L+1]$.
					\end{lemma}	
					\begin{remark}
						Note that here, each of the $	S_{k_1, k_2, k_3, l_j}$s is a subset of transmitters of size $\mu_T > 1$. 
					\end{remark}
					\begin{proof}
						Given that the set \eqref{scheduling_tau_geq1} has a size of 
						$K_R {M{\mu_T} \choose {{\mu_T}}}{{K_R-1} \choose {{\mu_R}}}{{K_R-{\mu_R}-1} \choose {{\mu_T} - 1}}$, representing the number of subfiles requiring delivery, it is sufficient to show that each receiver's ${M{\mu_T} \choose {{\mu_T}}}{{K_R-1} \choose {{\mu_R}}}{{K_R-{\mu_R}-1} \choose {{\mu_T} - 1}}$ requested subfiles are included. Since $\mathcal{R} \subseteq [K_R]\backslash\{1\} \text{ s.t }|\mathcal{R}| = {\mu_R}$, $\mathcal{T} \subseteq [K_R]\backslash \mathcal{R}\cup\{1\} \text{ s.t } |\mathcal{T}| = {\mu_T} -1$, and the set 
						$$
						\bigcup_{\substack{
								\mathcal{R} \subseteq [K_R]\backslash\{1\}:|\mathcal{R}| = {\mu_R}\\
								\mathcal{T} \subseteq [K_R]\backslash \mathcal{R}\cup\{1\}: |\mathcal{T}| = {\mu_T} -1	
						}}
						\Biggl\{
						W_{d_1, S_{k_1, 1, 1, l_1}, \mathcal{R}, \mathcal{T}} \Biggr\},
						$$
						consists ${{K_R - 1} \choose \mu_R}{{K_R -\mu_R - 1} \choose \mu_T-1}$ distinct subfiles requested by receiver $1$, and as
						$	S_{k_1, k_2, k_3, l_1} = S_{k_1, 1, 1, l_1} \oplus_{M} (k_2-1) + (k_3-1)M$ for $k_2 \in [2:M]$ and
						$k_3 \in [2:{M{\mu_T}-1 \choose {{\mu_T}-1}}]$, all the ${M{\mu_T} \choose {{\mu_T}}}{{K_R-1} \choose {{\mu_R}}}{{K_R-{\mu_R}-1} \choose {{\mu_T} - 1}}$ subfiles are included. 
						
						On the other hand, for each receiver $j \neq 1$, the set 
						$$
						\bigcup_{\substack{
								\mathcal{R} \subseteq [K_R]\backslash\{1\}:|\mathcal{R}| = {\mu_R}\\
								\mathcal{T} \subseteq [K_R]\backslash \mathcal{R}\cup\{1\}: |\mathcal{T}| = {\mu_T} -1	
						}}
						\Biggl\{\bigcup_{\substack{j \in \mathcal{R}}} W_{d_j, S_{k_1, 1, 1, l_1}, \{1\}\cup\mathcal{R}\backslash\{j\}, \mathcal{T}}
						\Biggr\},
						$$
						consists ${{K_R - 2} \choose {\mu_R-1}}{{K_R - \mu_R - 1} \choose {\mu_R-1}}$, the set 
						$$
						\bigcup_{\substack{
								\mathcal{R} \subseteq [K_R]\backslash\{1\}:|\mathcal{R}| = {\mu_R}\\
								\mathcal{T} \subseteq [K_R]\backslash \mathcal{R}\cup\{1\}: |\mathcal{T}| = {\mu_T} -1	
						}}
						\Biggl\{
						\bigcup_{\substack{j \in \mathcal{T}}} W_{d_j, S_{k_1, 1, 1, l_1}, \mathcal{R}, \{1\}\cup\mathcal{T}\backslash\{j\}}
						\Biggr\}
						$$
						consists ${{K_R - 2} \choose {\mu_R}}{{K_R - \mu_R-2} \choose {\mu_T-2}}$, and the set 
						$$
						\bigcup_{\substack{
								\mathcal{R} \subseteq [K_R]\backslash\{1\}:|\mathcal{R}| = {\mu_R}\\
								\mathcal{T} \subseteq [K_R]\backslash \mathcal{R}\cup\{1\}: |\mathcal{T}| = {\mu_T} -1	
						}}
						\Biggl\{ \bigcup_{\substack{j \in [K_R]\backslash\{1\}\cup\mathcal{R}\cup\mathcal{T}}} W_{d_j, S_{k_1,1, 1, l_j}, \mathcal{R}, \mathcal{T}}
						\Biggr\}
						$$
						consists  ${{K_R - 2} \choose {\mu_R}}{{K_R - \mu_R - 2} \choose {\mu_T-1}}$ distinct subfiles of the file requested by receiver $j$, and as 	$	S_{k_1, k_2, k_3, l_1} = S_{k_1, 1, 1, l_1} \oplus_{M} (k_2-1) + (k_3-1)M$ for $k_2 \in [2:M]$,
						$k_3 \in [2:{M{\mu_T}-1 \choose {{\mu_T}-1}}]$, and $j \in [L+1]$, all the ${M{\mu_T} \choose {{\mu_T}}}{{K_R-1} \choose {{\mu_R}}}{{K_R-{\mu_R}-1} \choose {{\mu_T} - 1}}$ subfiles are included. This is due to using Pascal's rule as:
						{ \begin{align*}
								{K_R-1 \choose \mu_R}{K_R-\mu_R-1 \choose \mu_T-1} & = {{K_R-2} \choose {\mu_R}}{{K_R-{\mu_R}-2} \choose {{\mu_T}-1}} + {{K_R-2} \choose {\mu_R}-1}	{{K_R-{\mu_R}-1} \choose {{\mu_T}-1}}\\ & +	{{K_R-2} \choose {\mu_R}}	{{K_R-{\mu_R}-2} \choose {{\mu_T}-2}},
						\end{align*}}
						and, as a result, all of the requested subfiles are partitioned in \eqref{scheduling_tau_geq1}.
					\end{proof}
					
					
					\begin{lemma}
						\label{lemma2_tau_geq1}
						At any communication block $k_1 \in [{K_R-1 \choose {\mu_R}}{K_R-{\mu_R}-1 \choose {\mu_T}-1}]$ from any super-block $k_2 \in [1:M]$ from any hyper-block $k_3 \in [1:{M\mu_T-1 \choose {\mu_T-1}}]$, 
						there exists a choice of beamforming coefficients, along with a network realization by IRS, such that the set of $K_R$ subfiles in 
						\begin{align}
							\Biggl\{&
							W_{d_1, S_{k_1, k_2, k_3, l_1}, \mathcal{R}, \mathcal{T}},
							\bigcup_{\substack{j \in \mathcal{R}}} W_{d_j, S_{k_1, k_2, k_3, l_1}, \{1\}\cup\mathcal{R}\backslash\{j\}, \mathcal{T}}, \nonumber \\
							&
							\bigcup_{\substack{j \in \mathcal{T}}} W_{d_j, S_{k_1, k_2, k_3, l_1}, \mathcal{R}, \{1\}\cup\mathcal{T}\backslash\{j\}},  \bigcup_{\substack{j \in [K_R]\backslash\{1\}\cup\mathcal{R}\cup\mathcal{T}}} W_{d_j, S_{k_1, k_2, k_3, l_j}, \mathcal{R}, \mathcal{T}}
							\Biggr\},
						\end{align}
						can be delivered interference-free 
						where at each communication block, $\mathcal{R}$ is a distinct subset of $[K_R]\backslash \{1\}$ of size $\mu_R$ and $\mathcal{T}$ is a distinct subset of $[K_R]\backslash \mathcal{R}\cup\{1\}$ of size $\mu_T-1$. 	
					\end{lemma}
					\begin{proof}
						Note that the selected subfiles in each communication block are either in the format introduced in Lemma \ref{lemma_1RX_without_interference} or \ref{lemma_tplustau_RX_without_interference}. Consequently, it is sufficient to set the beamforming coefficient $	{\tilde{{{\mathbf{V}}}}_{\mathcal{S}, \mathcal{R}, \mathcal{T}}^{[ji]}}(m)$ for each selected subfile accordingly. For the unselected subfiles in each communication block, the beamforming matrix should be set as a zero matrix. Then, we need an active IRS ‌with $Q = \mu_T L(L+1)$ elements such that at each communication block 		$k_1 \in \left[1:{K_R-1 \choose {\mu_R}}{K_R-{\mu_R}-1 \choose {{\mu_T}-1}}\right]$ from the super-block 		$k_2 \in [1:M]$ from the hyper-block 	$k_3 \in [1:{M{\mu_T}-1 \choose {{\mu_T}-1}}]$, the network matrix 	$\mathcal{N} =\left[n_{i,j}\right] \in {\{0, 1\}}^{\left\{M{\mu_T} \times K_R\right\}}$ is realized in the following way:
						\begin{itemize}
							\item From the transmitters 	$i \in S_{k_1, k_2, k_3, l_1}$ perspectives:
							\begin{equation}
								n_{i, j} = 
								\left\{\begin{array}{l}
									1, \quad \forall j \in {\mathcal{R}\cup\mathcal{T}\cup\{1\}}\\
									0, \quad \text{Otherwise}
								\end{array}\right.,
							\end{equation}
							\item From the transmitters 	$i \in S_{k_1, k_2, k_3, l_j}$ perspectives for 	$j \in [K_R]\backslash\{1\}\cup\mathcal{R}\cup\mathcal{T}$:
							\begin{equation}
								n_{i, k} = 
								\left\{\begin{array}{l}
									1, \quad \forall k \in {\mathcal{R}\cup\mathcal{T}\cup\{j\}}\\
									0, \quad \text{Otherwise}
								\end{array}\right.,
							\end{equation}
							\item And from the 	transmitters	$i \notin 	\left\{\bigcup_{\substack{j \in [1:L+1]}} S_{k_1, k_2, k_3, l_j}\right\}$ perspectives:
							\begin{equation}
								n_{i, k} = 1, \quad \forall
								k \in [K_R].
							\end{equation}
						\end{itemize}
						
						To have such a realization using an active IRS, we first denote 	$\tilde{\mathcal{N}}$ as 
						\begin{equation}
							\tilde{\mathcal{N}} = \left\{(i, j) \mid i \in [K_T], j \in [K_R], n_{i, j} = 0\right\}.
						\end{equation}
						
						Now, we need to configure each element of the IRS‌ in a way that at each communication block $t$, the following equations hold:
						\begin{equation}
							\sum_{u \in\{1, \ldots, Q\}} X^{[i]}(t) H_{\mathrm{TI}}^{[u i]}(t) H_{\mathrm{IR}}^{[j u]}(t) q^{[u]}(t)=-X^{[i]}(t) H^{[j i]}(t),\quad(i, j) \in \tilde{\mathcal{N}},
						\end{equation}
						where by omitting 	$X^{[i]}(t)$, they can be written as 
						\begin{equation}\label{RIS_crosslink_removal_1_2_3}
							\sum_{u \in\{1, \ldots, Q\}}  H_{\mathrm{TI}}^{[u i]}(t) H_{\mathrm{IR}}^{[j u]}(t) q^{[u]}(t)= -H^{[j i]}(t),\quad(i, j) \in \tilde{\mathcal{N}}.
						\end{equation}
						
						This equation denotes the removal of the cross-link between transmitter $i$ and receiver $j$ and has a solution almost surely \cite[Lemma 1]{bafghi_TCom}. Hence, $L = K_R - \mu_R - \mu_T$ unintended subfiles are not received at each receiver due to the undesired cross-links removal. Additionally, each receiver can cancel out $\mu_R$  subfiles due to its cached contents. Moreover, $\mu_T-1$ of the subfiles counted as outgoing interference for each receiver can be zero-forced due to transmitters' cooperation. These enable interference-free delivery of $K_R$ subfiles, each intended for a distinct receiver, at each communication block. Consequently, a DoF of $1$ is achievable at each block.
					\end{proof}
					\item Proof without Baranyai's theorem:
					If we do not use an existing $(M, \mu_T)$-subset partition given by Baranyai's theorem, we need to modify the prefetching phase slightly. Before explaining the necessary modification, let us state a lemma and define a symbol.
					\begin{lemma}
						The number of ordered partitions for the set $\{1, 2, \ldots, M\mu_T\}$ into $M$ subsets 	$A_1, A_2, \ldots, A_M$, each of size $\mu_T$, is 	
						$$
						{M{\mu_T} \choose {\mu_T}}
						{M{\mu_T}-{\mu_T} \choose {\mu_T}}
						\ldots
						{{\mu_T} \choose {\mu_T}}
						=
						\frac{(M{\mu_T})!}{({\mu_T}!)^{M}},
						$$
						where, for the ease of reference, we denote this number with
						$
						{M{\mu_T} \choose \underbrace{{\mu_T}, \ldots, {\mu_T}}_{\text{$M$}}}
						$. Moreover, denote such an ordered partition with 			$(A_1, A_2, \ldots, A_M) \vdash [M{\mu_T}]$.
						Note that $(A_1, A_2, A_3, \ldots, A_M) \vdash [M{\mu_T}]$
						and
						$(A_2, A_1, A_3, \ldots, A_M) \vdash [M{\mu_T}]$ are different.
						\begin{proof}
							To this end, the number of ways to choose the first subset of size $\mu_T$ is 	${M{\mu_T} \choose {\mu_T}}$. Then, the number of ways to choose the second subset of size $\mu_T$ is 			${M{\mu_T}-{\mu_T} \choose {\mu_T}}$ and so on. Eventually, the total number, using the rule of product, is attained. 
						\end{proof}
					\end{lemma} 
					Now, instead of breaking each file 		$W_k$
					into
					${M{\mu_T} \choose {\mu_T}}{K_R \choose {\mu_R}}$ distinct subfiles, we break it into 		$	{M{\mu_T} \choose \underbrace{{\mu_T}, \ldots, {\mu_T}}_{\text{$M$}}}{K_R \choose {\mu_R}}$ distinct subfiles. Consequently, the set of subfiles can be described as:
					\begin{align}\label{All_subfiles_format2}
						\mathcal{W} = 
						\big\{& W_{k, (A_1, \ldots, A_M), \mathcal{R}}: (A_1, \ldots, A_M) \vdash [M{\mu_T}], k \in [N], \nonumber\\ & |A_i| = {\mu_T}, i \in [M], \mathcal{R} \subseteq\left[K_{R}\right],|\mathcal{R}|={\mu_R} \big\}
					\end{align}
					
					Based on the above partitioning, each subfile 		$W_{k, (A_1, \ldots, A_M), \mathcal{R}}$ is cached at transmitter $i$, if 	$i \in A_1$. Therefore, each transmitter stores 		$N {M{\mu_T} -1 \choose {\mu_T}-1}{M{\mu_T} - {\mu_T} \choose {\mu_T}}\ldots{{\mu_T} \choose {\mu_T}}{K_R \choose {\mu_R}}$ subfiles in its cache memory, which is consistent with the size constraint as:
					$$
					{N}{M{\mu_T} -1 \choose {\mu_T}-1}{M{\mu_T} - {\mu_T} \choose {\mu_T}}\ldots{{\mu_T} \choose {\mu_T}}{K_R \choose {\mu_R}}\frac{F}{{M{\mu_T} \choose \underbrace{{\mu_T}, \ldots, {\mu_T}}_{\text{$M$}}}{K_R \choose {\mu_R}}} = M_TF \text{   { packets.}}
					$$
					
					Similarly, each subfile 		$W_{k, (A_1, \ldots, A_M), \mathcal{R}}$ is cached at receiver $j$, if 	$j \in \mathcal{R}$. Hence, each receiver stores $N {M{\mu_T} \choose \underbrace{{\mu_T}, \ldots, {\mu_T}}_{\text{$M$}}} {K_R-1 \choose {\mu_R}-1}$ subfiles in its cache memory, and its consistency with the size constraint can be easily verified as follows:
					$$
					{N}{{M{\mu_T}} \choose \underbrace{{\mu_T}, \ldots, {\mu_T}}_{\text{$M$}}}{{K_R-1} \choose {{\mu_R}-1}}\frac{F}{{M{\mu_T} \choose \underbrace{{\mu_T}, \ldots, {\mu_T}}_{\text{$M$}} }{K_R \choose {\mu_R}}} = M_RF \text{   { packets.}}
					$$		 
					
					In the delivery phase, based on the receiver's demand 	$\mathbf{d} = \left(W_{d_1}, W_{d_2}, \ldots, W_{d_{K_R}}\right)$, the transmitters should send the following 
					$
					K_{R} {M{\mu_T} \choose \underbrace{{\mu_T}, \ldots, {\mu_T}}_{\text{$M$}}} {K_{R}-1 \choose {\mu_R}}
					$
					subfiles:
					{
						\begin{align}\label{send_general_tau_new_break}
							\big\{&W_{d_{j}, (A_1, \ldots, A_M), \mathcal{R}}: j \in\left[K_{R}\right], (A_1, \ldots, A_M) \vdash [M{\mu_T}], \nonumber \\ & |A_i| = {\mu_T}, i \in [M], \mathcal{R} \subseteq\left[K_{R}\right]\backslash \{j\}, |\mathcal{R}|= {\mu_R}\big \}.
					\end{align}}
					In the following, we break each subfile 		$W_{d_j, (A_1, \ldots, A_M), \mathcal{R}}$ into 		${K_R-{\mu_R}-1} \choose {{\mu_T}-1}$ smaller subfiles, each being corresponded with a unique subset of receivers 		$\mathcal{R}\cup\{j\}\cup\mathcal{T}$, where 	$\mathcal{T} 
					$ represents the indices of an arbitrary set of size 		${\mu_T}-1$ from the 		$K_R-{\mu_R}-1$  receivers 		$[K_R]\backslash\mathcal{R}\cup\{j\}$. This offers the opportunity to eliminate each subfile not only in the receivers 		$\mathcal{R}$ due to prefetching, but also in the receivers 		$\mathcal{T}$ due to the possibility of cooperation between the transmitters.
					
					We assume the number of communication blocks as $H = {M{\mu_T} \choose \underbrace{{\mu_T}, \ldots, {\mu_T}}_{\text{$M$}}} {K_{R}-1 \choose {\mu_R}} {{K_R-{\mu_R}-1} \choose {{\mu_T} -1}}$. Here, for ease of reference, we term every $ {K_{R}-1 \choose {\mu_R}} {{K_R-{\mu_R}-1} \choose {{\mu_T} -1}}$ communication blocks a super-block, every $M$ super-blocks a hyper-block, and every $(M-1)!$ hyper-blocks a mega-block. Moreover, we assign numbers to 
					each of the ordered partitions 	$(A_1, \ldots, A_M) \vdash [M{\mu_T}]$, where 	$|A_i| = {\mu_T}$ for $i \in [M]$. 
					This is done in such a way that all the $M!$ partitions with the same composing subsets (different permutations) should be entirely and distinctly numbered with one of the sets 
					$[1:M!]$,
					$[M!+1:2 \times M!]$,
					$\ldots$,
					$\big[{M{\mu_T} \choose \underbrace{{\mu_T}, \ldots, {\mu_T}}_{\text{$M$}}}-M!+1:{M{\mu_T} \choose \underbrace{{\mu_T}, \ldots, {\mu_T}}_{\text{$M$}}} \big]$. Similar to the previous case, the number assigned to the ordered partition $S$ is denoted by $\kappa(S)$.
					
					Now, Lemmas \ref{absense_l1} and \ref{absense_l2} provide us with the achievable scheme in this case.
					
					\begin{lemma}\label{absense_l1}
						Given the prefetching phase for this case, for any receivers’ demand vector $\mathbf{d}$, the set of subfiles that need to be delivered can be partitioned into disjoint subsets of size $K_R$ as
						\begin{align}
							\bigcup_{\substack{
									k_2 \in \left[1:M\right]\\ 
									k_3 \in [1:(M-1)!]\\
									k_4 \in \left[1:\frac{1}{M!}{M{\mu_T} \choose {\underbrace{{\mu_T}, \ldots, {\mu_T}}_{\text{$M$}}}}\right]
									\\
									\mathcal{R} \subseteq [K_R]\backslash\{1\}:|\mathcal{R}| = {\mu_R}\\
									\mathcal{T} \subseteq [K_R]\backslash \mathcal{R}\cup\{1\}: |\mathcal{T}| = {\mu_T} -1	
							}}
							\Biggl\{&
							W_{d_1, S_{k_1, k_2, k_3, k_4, l_1}, \mathcal{R}, \mathcal{T}},
							\bigcup_{\substack{j \in \mathcal{R}}} W_{d_j, S_{k_1, k_2, k_3, k_4, l_1}, \{1\}\cup\mathcal{R}\backslash\{j\}, \mathcal{T}}, \nonumber \\
							&
							\bigcup_{\substack{j \in \mathcal{T}}} W_{d_j, S_{k_1, k_2, k_3, k_4, l_1}, \mathcal{R}, \{1\}\cup\mathcal{T}\backslash\{j\}},  \bigcup_{\substack{j \in [K_R]\backslash\{1\}\cup\mathcal{R}\cup\mathcal{T}}} W_{d_j, S_{k_1, k_2, k_3, k_4, l_j}, \mathcal{R}, \mathcal{T}}
							\Biggr\},
						\end{align}
						where $\left\{S_{k_1, 1, 1, 1, l_1}, S_{k_1, 1, 1, 1, l_2}, \ldots, S_{k_1, 1, 1, 1, l_{L+1}}\right\}$ denotes arbitrary distinct indices from an arbitrary set in the form 
						\begin{equation}\label{desired_format_set_factorial}
							\left[k\times(M-1)!+1:(k+1)\times(M-1)!\right],\quad k \in [0:M-1],
						\end{equation}
						for each  $k_1 \in \left[1:{K_R-1 \choose {\mu_R}}{K_R-\mu_R-1 \choose {\mu_T-1}}\right]$, and we have:
						{\begin{align}
								& \kappa(S_{k_1, k_2, k_3, k_4, l}) = \kappa(S_{k_1, 1, 1, 1, l}) \oplus_{M!} (k_2-1)(M-1)!+(k_3-1)M+(k_4-1)M!, \nonumber\\ & (k_2, k_3, k_4) \neq (1, 1, 1),
						\end{align}}
						for
						$k_1 \in \left[1:{K_R-1 \choose {\mu_R}}{K_R-{\mu_R}-1 \choose {{\mu_T}-1}}\right]$, 
						$k_2 \in [1:M]$,
						$k_3 \in [1:(M-1)!]$,
						$k_4 \in \big[1:\frac{1}{M!}{M{\mu_T} \choose {\underbrace{{\mu_T}, \ldots, {\mu_T}}_{\text{$M$}}}}\big]$, and $j \in [L+1]$.
					\end{lemma}
					\begin{remark}
						Note that $S_{k_1, k_2, k_3, k_4, l_j}$ is an ordered partition where each of its composing subsets has a size of $\mu_T$.
					\end{remark}
					\begin{proof}
						The proof is very similar to the proof of Lemma \ref{lemma1_tau_geq1}.
					\end{proof}
					\begin{lemma}\label{absense_l2}
						At any communication block $k_1 \in [{K_R-1 \choose {\mu_R}}{K_R-{\mu_R}-1 \choose {\mu_T}-1}]$ from any super-block $k_2 \in [1:M]$ from any hyper-block $k_3 \in [1:(M-1)!]$ from any mega-block $k_4 \in \big[1:\frac{1}{M!}{M{\mu_T} \choose {\underbrace{{\mu_T}, \ldots, {\mu_T}}_{\text{$M$}}}}\big]$, and $j \in [L+1]$, 
						there exists a choice of beamforming coefficients, along with a network realization by IRS, such that the set of $K_R$ subfiles in 
						\begin{align}
							\Biggl\{&
							W_{d_1, S_{k_1, k_2, k_3, k_4, l_1}, \mathcal{R}, \mathcal{T}},
							\bigcup_{\substack{j \in \mathcal{R}}} W_{d_j, S_{k_1, k_2, k_3, k_4, l_1}, \{1\}\cup\mathcal{R}\backslash\{j\}, \mathcal{T}}, \nonumber \\
							&
							\bigcup_{\substack{j \in \mathcal{T}}} W_{d_j, S_{k_1, k_2, k_3, k_4, l_1}, \mathcal{R}, \{1\}\cup\mathcal{T}\backslash\{j\}},  \bigcup_{\substack{j \in [K_R]\backslash\{1\}\cup\mathcal{R}\cup\mathcal{T}}} W_{d_j, S_{k_1, k_2, k_3, k_4, l_j}, \mathcal{R}, \mathcal{T}}
							\Biggr\},
						\end{align}
						can be delivered interference-free 
						where at each communication block, $\mathcal{R}$ is a distinct subset of $[K_R]\backslash \{1\}$ of size $\mu_R$ and $\mathcal{T}$ is a distinct subset of $[K_R]\backslash \mathcal{R}\cup\{1\}$ of size $\mu_T-1$. 	
					\end{lemma}
					\begin{proof}
						The proof is very similar to the proof of Lemma \ref{lemma2_tau_geq1}.
					\end{proof}
				\end{enumerate}
				\item $\mu_R+\mu_T+L < K_R$:
				In this setting, similar to case
				(I), we consider the following two cases depending on whether Baranyai's theorem is used or not:
				\begin{enumerate}
					[wide, labelwidth=!, labelindent=0pt]
					\item Using Baranyai's theorem:
					
                    To this end, we break each subfile 				$W_{d_j, \mathcal{S}, \mathcal{R}, \mathcal{T}}$ into 		${{K_R-({\mu_R}+{\mu_T})} \choose L}$ smaller subfiles 		$W_{d_j, \mathcal{S}, \mathcal{R}, \mathcal{T}, \mathcal{L}}$,
					each being corresponding with a unique subset of receivers 
					$\mathcal{R}\cup\{j\}\cup\mathcal{T}\cup\mathcal{L}$ where 		$\mathcal{L} = \{\mathcal{R}_1, \ldots, \mathcal{R}_L\}$ is an arbitrary subset of size $L$‌ from the receivers $[K_R]\backslash\mathcal{R}\cup\mathcal{T}\cup\{j\}$. We, then, consider the number of communication blocks to be $H = {M{\mu_T} \choose {\mu_T}}{{K_R} \choose {{\mu_R}+{\mu_T}+L}}{{{\mu_R}+{\mu_T}+L-1} \choose {\mu_R}}{{{\mu_T}+L-1} \choose {{\mu_T}-1}}$. Here, for ease of reference, we term every 	${{{\mu_R}+{\mu_T}+L-1} \choose {\mu_R}}{{{\mu_T}+L-1} \choose {{\mu_T}-1}}$ communication blocks a super-block and every 	$M{\mu_T} \choose {\mu_T}$ super-blocks a hyper-block. In the following, we consider a subset $\mathcal{C} \subseteq [K_R]$ of size 		${\mu_R}+{\mu_T}+L$. Then, for the sub-network of 		$[K_T] \times \mathcal{C}$, we configure the IRS‌ as explained in the case I.A. It is worth mentioning that the sub-network of 		$[K_R]\backslash\mathcal{C}$ is supposed to be fully-connected, and, as a result, there is no need for further IRS‌ configuration.
					Therefore, the number
					of IRS elements needs to be $Q = \mu_T L(L + 1)$. 
					
					Now, it is easy to see that the DoF of $1$ is achievable for each of the receivers in the subset
					$\mathcal{C}$ in 		${M{\mu_T} \choose {\mu_T}} {{{\mu_R}+{\mu_T}+L-1} \choose {\mu_R}}{{\mu_T}+L-1 \choose {{\mu_T}-1}}$ communication blocks of a hyper-block, and the DoF of all the receivers
					$[KR]\backslash\mathcal{C}$ is equal to zero. Subsequently, it is sufficient to repeat this procedure for other subsets
					of size 	${\mu_T}+{\mu_R}+L$ from the receivers. To calculate the overall DoF, note that each receiver
					can decode its desired subfile interference-free in 	${M{\mu_T} \choose {\mu_T}}{{K_R-1} \choose {{\mu_R}+{\mu_T}+L-1}}{{{\mu_R}+{\mu_T}+L-1} \choose {{\mu_R}}}{{{\mu_T}+L-1} \choose {{\mu_T}-1}}$ blocks
					and
					the number of all blocks is
					${M{\mu_T} \choose {\mu_T}}{{K_R} \choose {{\mu_R}+{\mu_T}+L}}{{{\mu_R}+{\mu_T}+L-1} \choose {{\mu_R}}}{{{\mu_T}+L-1} \choose {{\mu_T}-1}}$. Consequently, the DoF of
					$$\mathrm{DoF_{user}} = \frac{{M{\mu_T} \choose {\mu_T}}{{K_R-1} \choose {{\mu_R}+{\mu_T}+L-1}}{{{\mu_R}+{\mu_T}+L-1} \choose {{\mu_R}}}{{{\mu_T}+L-1} \choose {{\mu_T}-1}}}{{M{\mu_T} \choose {\mu_T}}{{K_R} \choose {{\mu_R}+{\mu_T}+L}}{{{\mu_R}+{\mu_T}+L-1} \choose {{\mu_R}}}{{{\mu_T}+L-1} \choose {{\mu_T}-1}}} = \frac{{\mu_R}+{\mu_T}+L}{K_R},$$
					is achievable for each receiver, and the sum-DoF of ${\mu_R}+{\mu_T}+L$ is achievable in this case. 
					\item Without Baranyai's theorem:
					
					In this case, we break each subfile 					$W_{d_j, (A_1, \ldots, A_M), \mathcal{R}, \mathcal{T}}$ into 		${{K_R-({\mu_R}+{\mu_T})} \choose L}$ smaller subfiles 		$W_{d_j, (A_1, \ldots, A_M), \mathcal{R}, \mathcal{T}, \mathcal{L}}$,
					each being corresponding with a unique subset of receivers 
					$\mathcal{R}\cup\{j\}\cup\mathcal{T}\cup\mathcal{L}$ where 		$\mathcal{L} = \{\mathcal{R}_1, \ldots, \mathcal{R}_L\}$ is an arbitrary subset of size $L$‌ from the receivers $[K_R]\backslash\mathcal{R}\cup\mathcal{T}\cup\{j\}$. The number of blocks is considered to be $H = {M{\mu_T} \choose {\underbrace{{\mu_T}, \ldots, {\mu_T}}_{\text{$M$}}}}{{K_R} \choose {{\mu_R}+{\mu_T}+L}}{{{\mu_R}+{\mu_T}+L-1} \choose {\mu_R}}{{{\mu_T}+L-1} \choose {\mu_T-1}}$. Here, for the ease of reference, we term every 		${{{\mu_R}+{\mu_T}+L-1} \choose {\mu_R}}{{{\mu_T}+L-1} \choose {{\mu_T}-1}}$ communication blocks a super-block and every 	${M{\mu_T} \choose {\underbrace{{\mu_T}, \ldots, {\mu_T}}_{\text{$M$}}}}$ super-blocks a hyper-block. In the following, we consider a subset $\mathcal{C} \subseteq [K_R]$ of size 		${\mu_R}+{\mu_T}+L$. Then, for the sub-network of 		$[K_T] \times \mathcal{C}$, we configure the IRS‌ as explained in the case I.B. It is worth mentioning that the sub-network of 		$[K_R]\backslash\mathcal{C}$ is supposed to be fully-connected, and, as a result, there is no need for further IRS‌ configuration.
					Therefore, the number
					of IRS elements needs to be $Q = \mu_T L(L + 1)$. 
					
					Now, it is easy to see that the DoF of $1$ is achievable for each of the receivers in the subset
					$\mathcal{C}$ in 		${M{\mu_T} \choose {\underbrace{{\mu_T}, \ldots, {\mu_T}}_{\text{$M$}}}} {{{\mu_R}+{\mu_T}+L-1} \choose {\mu_R}}{{\mu_T}+L-1 \choose {{\mu_T}-1}}$ communication blocks of a hyper-block, and the DoF of all the other receivers
					$[KR]\backslash\mathcal{C}$ is equal to zero. Subsequently, it is sufficient to repeat this procedure for other subsets
					of size 	${\mu_T}+{\mu_R}+L$ from the receivers. To calculate overall DoF, note that each receiver
					can decode its desired subfile interference-free in 	${M{\mu_T} \choose {\underbrace{{\mu_T}, \ldots, {\mu_T}}_{\text{$M$}}}}{{K_R-1} \choose {{\mu_R}+{\mu_T}+L-1}}{{{\mu_R}+{\mu_T}+L-1} \choose {{\mu_R}}}{{{\mu_T}+L-1} \choose {{\mu_T}-1}}$ blocks
					and he number of all communication blocks is
					${M{\mu_T} \choose {\underbrace{{\mu_T}, \ldots, {\mu_T}}_{\text{$M$}}}}{{K_R} \choose {{\mu_R}+{\mu_T}+L}}{{{\mu_R}+{\mu_T}+L-1} \choose {{\mu_R}}}{{{\mu_T}+L-1} \choose {{\mu_T}-1}}$. Consequently, the DoF of
					$$\mathrm{DoF_{user}} = \frac{{M{\mu_T} \choose {\underbrace{{\mu_T}, \ldots, {\mu_T}}_{\text{$M$}}}}{{K_R-1} \choose {{\mu_R}+{\mu_T}+L-1}}{{{\mu_R}+{\mu_T}+L-1} \choose {{\mu_R}}}{{{\mu_T}+L-1} \choose {{\mu_T}-1}}}{{M{\mu_T} \choose {\underbrace{{\mu_T}, \ldots, {\mu_T}}_{\text{$M$}}}}{{K_R} \choose {{\mu_R}+{\mu_T}+L}}{{{\mu_R}+{\mu_T}+L-1} \choose {t}}{{{\mu_T}+L-1} \choose {{\mu_T}-1}}} = \frac{{\mu_R}+{\mu_T}+L}{K_R},$$
					is achievable for each receiver, and the sum-DoF of ${\mu_R}+{\mu_T}+L$ is achievable in this case. 
				\end{enumerate}
			\end{enumerate}
			
			\bibliography{refs}
			\bibliographystyle{IEEEtran}
		\end{document}